\documentclass[journal, letter, 12pt, onecolumn, draftclsnofoot]{IEEEtran}

\usepackage[utf8]{inputenc} %
\usepackage[T1]{fontenc}    %
\usepackage{hyperref}       %
\usepackage{url}            %
\usepackage{booktabs}       %
\usepackage{nicefrac}       %
\usepackage{microtype}      %
\usepackage{amssymb,amsfonts,amsmath,amsthm,amscd,dsfont,mathrsfs,bbm}
\usepackage{graphicx,float,psfrag,pstcol,epsfig,color,subcaption}%
\usepackage{psfrag}
\usepackage{pstcol}
\usepackage{import}
\usepackage{cite}
\usepackage{setspace}
\usepackage{url}
\usepackage{booktabs}
\usepackage{algorithm,algpseudocode}
\usepackage{bm} %

\usepackage{tikz}%
\usetikzlibrary{shapes, positioning} %
\usetikzlibrary{pgfplots.groupplots}
\usetikzlibrary{arrows}
\usetikzlibrary{positioning}
\usetikzlibrary{fit,backgrounds}
\usepackage{pgfplots}%

\allowdisplaybreaks

\newcommand{\X}{{\mathcal{X}}}
\newcommand{\Y}{{\mathcal{Y}}}

\newcommand{\cX}{{\cal X}}
\newcommand{\cY}{{\cal Y}}
\newcommand{\cU}{{\cal U}}
\newcommand{\cI}{{\cal I}}
\newcommand{\cJ}{{\cal J}}

\newcommand{\ux}{\underline{x}}

\newcommand{\uw}{\underline{w}}

\newcommand{\cS}{{\cal S}}

\newcommand{\E}{\mathbb{E}}

\newcommand{\dtm}{\mathbf{B}}
\newcommand{\dtmh}{\hat{\dtm}}

\newcommand{\dtmt}{\tilde{\dtm}}

\newcommand{\T}{\mathrm{T}}
\newcommand{\psib}{\bm{\psi}}
\newcommand{\eps}{\epsilon}
\newcommand{\uf}{\underline{f}}

\DeclareMathOperator{\tr}{tr}

\newcommand{\bxi}{\boldsymbol{\xi}}
\newcommand{\bXi}{\boldsymbol{\Xi}}
\newcommand{\bQ}{\mathbf{Q}}
\newcommand{\bphi}{\boldsymbol{\phi}}
\newcommand{\V}{{\mathcal{V}}}
\newcommand{\C}{{\mathcal{C}}}
\newcommand{\cN}{\mathcal{N}}
\DeclareMathOperator*{\relint}{relint}
\newcommand{\cP}{\mathcal{P}}
\newcommand{\frob}[1]{\|#1\|_\mathrm{F}}
\newcommand{\bfrob}[1]{\bigl\|#1\bigr\|_\mathrm{F}}
\newcommand{\bbfrob}[1]{\left\|#1\right\|_\mathrm{F}}

\newcommand{\spectral}[1]{\|#1\|_\mathrm{s}}

\newcommand{\bbspectral}[1]{\left\|#1\right\|_\mathrm{s}}
\newcommand{\trop}[1]{\tr\left\{#1\right\}}
\newcommand{\defeq}{\triangleq}

\newcommand{\simpX}{\cP^\X}

\newcommand{\nbhd}{\cN}

\newcommand{\Ph}{\hat{P}}
\DeclareMathOperator{\1}{\mathds{1}} %
\newcommand{\Gb}{\mathbf{G}}
\newcommand{\Wb}{\mathbf{W}}
\newcommand{\Ib}{\mathbf{I}}
\newcommand{\Kb}{\mathbf{K}}
\newcommand{\Lb}{\mathbf{L}}
\newcommand{\Lbh}{\hat{\Lb}}
\newcommand{\Cb}{\mathbf{C}}

\newcommand{\Ab}{\mathbf{A}}
\newcommand{\Bh}{\hat{B}}
\newcommand{\Mb}{\mathbf{M}}
\newcommand{\Ub}{\mathbf{U}}
\newcommand{\eb}{\bm{e}}
\newcommand{\ub}{\bm{u}}
\newcommand{\vb}{\bm{v}}
\newcommand{\Pb}{\mathbf{P}}
\newcommand{\Qb}{\mathbf{Q}}
\newcommand{\cH}{\mathcal{H}}
\newcommand{\fh}{\hat{f}}
\newcommand{\zetab}{\bm{\zeta}}
\newcommand{\phib}{\bm{\phi}}
\newcommand{\Phib}{\mathbf{\Phi}}
\newcommand{\Psib}{\mathbf{\Psi}}
\newcommand{\Psibh}{\hat{\Psib}}
\newcommand{\psibh}{\hat{\psib}}
\newcommand{\psibt}{\tilde{\psib}}
\newcommand{\psih}{\hat{\psi}}
\newcommand{\cL}{\mathcal{L}}
\newcommand{\Xib}{\bm{\Xi}}
\newcommand{\dtmdmat}{\Xib}

\newcommand{\Jb}{\mathbf{J}}
\newcommand{\xib}{\bm{\xi}}

\newcommand{\zerob}{\bm{0}}
\DeclareMathOperator{\vecop}{vec}
\DeclareMathOperator{\cov}{cov}
\DeclareMathOperator*{\argmax}{arg\,max}
\DeclareMathOperator{\symdif}{\triangle}

\theoremstyle{definition}
\newtheorem{definition}{Definition}%
\theoremstyle{plain}
\newtheorem{theorem}{Theorem}%
\theoremstyle{plain}
\newtheorem{prop}{Proposition}%
\theoremstyle{plain}
\newtheorem{lemma}{Lemma}%
\theoremstyle{plain}
\theoremstyle{plain}
\theoremstyle{plain}
\theoremstyle{remark}
\newtheorem{example}{Example}
\theoremstyle{property}
\newtheorem{property}{Property}

\title{An Information-theoretic Approach to Unsupervised Feature Selection for High-Dimensional Data}

\author{
Shao-Lun Huang,~\IEEEmembership{Member,~IEEE,}
Xiangxiang Xu,~\IEEEmembership{Student Member,~IEEE,}
and~Lizhong Zheng,~\IEEEmembership{Fellow,~IEEE}%
\thanks{S.-L. Huang is with the Data Science and Information Technology Research Center, Tsinghua-Berkeley Shenzhen Institute, Shenzhen 518055, China (e-mail:
shaolun.huang@sz.tsinghua.edu.cn).}%
\thanks{X. Xu is with the Department of Electronic Engineering, Tsinghua University, Beijing 100084, China (e-mail: xuxx14@mails.tsinghua.edu.cn).}%
\thanks{L. Zheng is with the Department of Electrical Engineering and Computer Science, Massachusetts Institute of Technology, Cambridge, MA 02139, USA (e-mail: lizhong@mit.edu).}
}

\def\secsc{0}

\begin{document}

\maketitle

\begin{abstract} 
  In this paper, we propose an information-theoretic approach to design the functional representations to extract the hidden common structure shared by a set of random variables. The main idea is to measure the common information between the random variables by Watanabe's total correlation, and then find the hidden attributes of these random variables such that the common information is reduced the most given these attributes. We show that these attributes can be characterized by an exponential family specified by the eigen-decomposition of some pairwise joint distribution matrix. Then, we adopt the log-likelihood functions for estimating these attributes as the desired functional representations of the random variables, and show that such representations are informative to describe the common structure. Moreover, we design both the multivariate alternating conditional expectation (MACE) algorithm to compute the proposed functional representations for discrete data, and a novel neural network training approach for continuous or high-dimensional data. 
\if\secsc1{In addition, we analyze the sample complexity of the MACE algorithm, which indicates the number of samples needed to train the functional representations well. }\fi
Furthermore, we show that our approach has deep connections to existing techniques, such as Hirschfeld-Gebelein-R\'{e}nyi (HGR) maximal correlation, linear principal component analysis (PCA), and consistent functional map, which establishes insightful connections between information theory and machine learning. Finally, the performances of our algorithms are validated by numerical simulations.

\end{abstract}

\IEEEpeerreviewmaketitle

\section{Introduction}

Given a set of $d$ discrete random variables $X^d = (X_1, \ldots, X_d)$ with the (unknown) joint distribution $P_{X^d}$, and a sequence of observed sample vectors $\ux^{(\ell)} = (x_1^{(\ell)} , \ldots , x_d^{(\ell)})$ i.i.d. generated from this joint distribution, for $\ell = 1, \ldots , n$, our goal in this paper is to efficiently and effectively extract the hidden common information structure (or simply called common structure) shared by these random variables from the observed sample vectors. This is a typical unsupervised learning problem, and such common structures can be useful in many machine learning scenarios. As a motivating example, in the MNIST digits recognition problem~\cite{MNIST}, we often divide the images into overlapping sub-images, such as in Fig.~\ref{fig:MNIST}, and then train feature functions on the sub-images for learning the digits. %
In this problem, we can view each sub-image as a random variable $X_i$, and the training images as the observed data vectors. Since these sub-images are constructed from the written digits, the digit is the key common information shared by these sub-images. Therefore, effectively mining the information of the shared structure among these random variables can be helpful for recognizing the digits.

\begin{figure}
\centering
\includegraphics[width=.5\textwidth]{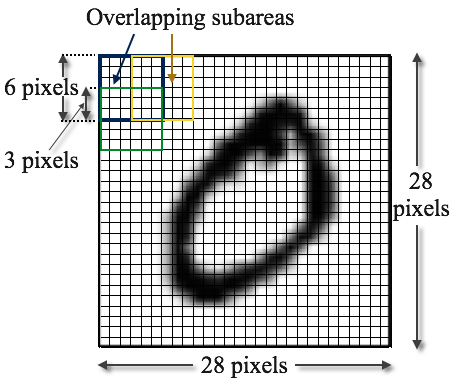}
\caption{The division of images into $8 \times 8 = 64$ overlapping subareas. Each subarea has $6 \times 6$ pixels, and nearby subareas overlap with $3$ pixels.}
\label{fig:MNIST}
\end{figure}

In addition, the concept of extracting common structure shared by multiple random variables or objects has also appeared or implicitly posted in several disciplines. For instance, linear principal component analysis (PCA)~\cite{Jolliffe}, the most widely adopted unsupervised learning technique, can be viewed as resolving a principal direction that conveys the most common randomness among different dimensions of data vectors. In addition, the consistent functional map network~\cite{Huang:2014, Wang:2013, Wang:2014}, a recently proposed effective approach in computer vision, takes each $X_i$ as a shape, and aims to find the shared components among different shapes. The main issue behind these problems is: how to design good low-dimensional functions of the random variables $X^d$, such that these functional representations are effective to reveal the common structure among these random variables. This can also be viewed as the unsupervised dimension reduction problem with the particular focus on extracting the common information of random variables.
In this paper, our goal is to apply the ideas from information theory to design good algorithms for finding such useful functional representations.

Our approach can be delineated in the following steps. Firstly, we want to identify the targeted random variable $U$ embedded in the random variables $X^d$ with some joint distribution $P_{U X^d}$, such that $U$ contains much information about the common structure shared by $X^d$. For this purpose, we apply the Watanabe's total correlation (or simply called the total correlation~\cite{watanabe}) to measure the amount of information shared by multiple random variables, and then find the optimal embedded $U$ such that the reduction of the total correlation given the knowledge of $U$ is maximized. To extract the effective low-dimensional features, we restrict the information volume of $U$ about $X_1, \ldots, X_d$ to be small, so that we can concentrate on the most ``learnable'' part of information about the common structures from the data. We show that in this small information rate regime of $U$, the optimal embedded $U$ can be characterized by an exponential family induced by the largest eigenvector of a pairwise joint distribution matrix. Then, we apply the log-likelihood function of estimating $U$ from $X_1, \ldots, X_d$ in this exponential family as the functional representation for extracting the common structure. Since $U$ is informative about the common structure and the log-likelihood function is the sufficient statistic of the observed data vectors about the target $U$, such a functional representation is effective to extract the common structure shared by these random variables. In addition, we extend this approach to searching for a sequence of mutually independent random variables $U^k = (U_1, \dots, U_k)$, such that the reduction of the total correlation is maximized. It turns out that the log-likelihood functions for estimating $U^k$ precisely correspond to the top $k$ eigenvectors of the pairwise joint distribution matrix, which establishes a decomposition of the common information between multiple random variables to principal modes of the pairwise joint distribution matrix.

Moreover, we demonstrate that these functional representations can be directly computed from the observed data vectors by a multivariate alternating conditional expectation (MACE) algorithm, which generalizes the traditional ACE algorithm~\cite{Breiman} to more than two random variables. This offers an efficient and reliable way to compute useful functional representations from discrete data variables. Furthermore, for high-dimensional or continuous data variables such that the conditional expectations are hardly accurately estimated from the limited data samples, we show that the functional representations can be computed through neural networks by optimizing a pairwise correlation loss. This offers a novel neural network training architecture for jointly analyzing multi-modal data. 
\if\secsc1{We also investigate the sample complexity of the MACE algorithm for i.i.d. sampled training data, by characterizing the error exponent of the failure probability of the MACE algorithm with respect to a given learning error. Our results indicate the number of training samples required to learn these functional representations in the sense of the generalization error.}
\fi

Finally, we show in Section~\ref{sec:5} that our approach shares deep connections and can be viewed as generalizations to several existing techniques, including the Hirschfeld-Gebelein-R\'{e}nyi (HGR) maximal correlation~\cite{renyi}, linear PCA, and consistent functional map network. This combines the knowledge from different domains, and offers a unified understanding for disciplines in information theory, statistics, and machine learning. We would also like to mention that the idea of studying the tradeoff between the total correlation and the common information rate was also employed in~\cite{op2016caching}\cite{op2016total} for Gaussian vectors in caching problems, while our works investigate this tradeoff for general discrete random variables. Moreover, the correlation explanation (CorEx) introduced by~\cite{ttc} also applied the total correlation as the information criterion to unsupervised learning. In particular, the authors in~\cite{ttc} solved an optimization problem by restricting the cardinality of $U$, and a rather complicated iterative algorithm was derived. On the other hand, in this paper we restrict the information volume contained in $U$, which is a more natural constraint in information theory, and we obtain clean analytical solutions that can be computed by simple and efficient algorithms.

In the rest of this paper, we introduce the details of our information theoretic approach for extracting the common structure via functional representations, and present the resulted algorithm design, 
\if\secsc1{the sample complexity of the algorithm, }\fi
 as well as their applications to practical problems.

\section{The Information Theoretic Approach} \label{sec:2}

\begin{figure}
\centering 
\includegraphics[width=0.5\textwidth]{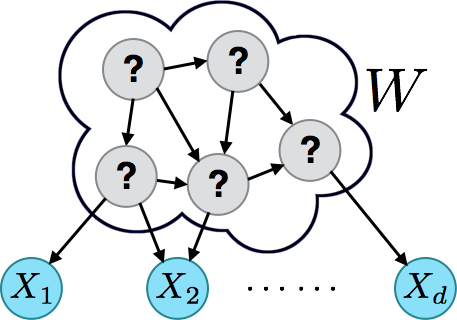} 
\caption{The random variables $X_1, \ldots , X_d$ are conditional independently generated from some hidden structure $W$.}
\label{fig:figure_1}
\end{figure}

Given discrete random variables $X^d = (X_1, \ldots, X_d)$ with the ranges $\cX^d \defeq \cX_1\times \dots \times \cX_d$ and joint distribution $P_{X^d}$, we model the common structure shared by these random variables as a high-dimensional latent variable $W$, in which the random variables $X_1 , \ldots , X_d$ are conditionally independent given $W$, i.e., $P_{X^d | W} = \Pi_{i = 1}^d P_{X_i | W}$, as depicted in Fig.~\ref{fig:figure_1}. Our goal is to learn this common structure from i.i.d. sample vectors generated from $P_{X^d}$. Since the correlation between $W$ and $X_i$'s is generally complicated, it is difficult to directly identify and learn the structures of $W$ without the labels and assumptions on the generating models of $X_i$'s as in the unsupervised learning scenarios. Therefore, instead of identifying the high-dimensional latent variable $W$, we focus on learning the low-dimensional random variable $U$ that contains much common information shared between $X_i$'s, which can be viewed as the informative attribute for the common structure. 

To identify such variable, we apply the total correlation\footnote{Specifically, for random variables $X_1, \ldots , X_d$, the total correlation is defined as the Kullback-Leibler (K-L) divergence $D(P_{X_1 \cdots X_d} \| P_{X_1} \cdots P_{X_d})$ between the joint distribution and the product of the marginal distributions.} \cite{watanabe} to measure the amount of common information shared between multiple random variables. Then, the amount of information that an attribute $U$ contains about the common structure shared by the random variables $X_i$'s is measured by the reduction of the total correlation given the knowledge of $U$, defined as
\begin{align} %
&{\cal L} ( X^d  | U) %
  \triangleq D(P_{X^d} \| P_{X_1} \cdots P_{X_d}) - D(P_{X^d} \| P_{X_1} \cdots P_{X_d} | U).\label{eq:total_correlation}
\end{align}
Our goal is to identify the targeted random variable $U$ with the information rate constraint\footnote{Note that $I(U; X^d)$ measures the amount of information of $U$ about the whole $X^d$, while ${\cal L} ( X^d  | U)$ measure the amount of information only about the common structure. The constraint $I(U; X^d) \leq \delta$ allows us to focus on low-dimensional attribute of $W$, in which we typically choose $\delta$ to be small.} $I(U; X^d) \leq \delta$, for some given $\delta$, such that the reduction of the total correlation is maximized. This can be formulated as the optimization problem:
\begin{subequations}
  \begin{alignat}{2} \label{eq:total_correlation2a}
    &\quad\max_{P_{U X^d }}&~ & \cL( X^d | U),  \\  
    &\text{subject to:} & &I(U; X^d) \leq \delta.
  \end{alignat}
  \label{eq:total_correlation2}
\end{subequations}
In particular, we would like to focus on the low-rate regime of $U$, %
which assumes $\delta$ to be small. This allows us to concentrate on the most representative low-dimensional attribute for describing the common structure. In addition, we make an extra constraint that 
\begin{align} \label{eq:no_tiltied}
  \min_u P_{U}(u) > \gamma,
\end{align}
for some finite $\gamma > 0$ irrelevant to $\delta$, which is natural for many machine learning problems. While the optimization problem~\eqref{eq:total_correlation2} in general has no analytical solution, in the regime of small $\delta$, it can be solved by a local information geometric approach, in which the optimal solutions can be specified by the eigen-decomposition of some pairwise joint distribution matrix.

\subsection{The Local Information Geometry}

To delineate our approach and results, we define the matrix $\dtm$ from the pairwise joint distributions as
\begin{align} \label{eq:Bmatrix}
\dtm = 
\left[
\begin{array}{cccc}
\mathbf{I}_{(1)} & \dtm_{12} & \cdots & \dtm_{1d} \\
\dtm_{21} & \mathbf{I}_{(2)} & \cdots & \dtm_{2d} \\
\vdots & \vdots & \ddots & \vdots \\
\dtm_{d1} & \dtm_{d2}  & \cdots & \mathbf{I}_{(d)}
\end{array} 
\right]
\end{align}
where $\mathbf{I}_{(i)}$ are $| \cX_i | \times | \cX_i |$ identity matrices, for all $i$, and $\dtm_{ij}$ are $( | \cX_i | \times | \cX_j | )$-dimensional matrices with the entry in the $x_i$-th row and $x_j$-th column defined as, for all $i \neq j$,
\begin{align*}%
B_{ij} (x_i; x_j) = \frac{P_{X_i X_j}(x_i , x_j)}{\sqrt{P_{X_i}(x_i)}\sqrt{P_{X_j}(x_j)}}. %
\end{align*}
The eigen-decomposition of the matrix $\dtm$ has the following properties.
\begin{lemma} \label{lem:1}
Let the eigenvalues and eigenvectors of the matrix $\dtm$ be $\lambda^{(0)} \geq \lambda^{(1)} \geq \cdots \geq \lambda^{(m-1)}$ and $\psib^{(0)}, \psib^{(1)}, \ldots , \psib^{(m-1)}$, respectively, where $m \triangleq \sum_{i = 1}^d |\cX_i|$ is the dimensionality of $\dtm$. In addition, let $\vb_i$ be the $|\cX_i|$-dimensional vector such that $v_i(x_i) = \sqrt{P_{X_i}(x_i)}$, then
\begin{enumerate}
\item $\dtm$ is a positive semidefinite matrix, i.e., $\lambda^{(m-1)} \geq 0$.
\item The largest eigenvalue $\lambda^{(0)} = d$ with the corresponding eigenvector $\psib^{(0)} = \frac{1}{\sqrt{d}} \left[\vb_1^{\T} , \ldots , \vb_d^{\T}\right]^{\T}$.  
\item The second largest eigenvalue $\lambda^{(1)} \geq 1$.
\item The last $d-1$ eigenvalues $\lambda^{(m-d+1)} = \cdots =\lambda^{(m-1)} = 0$, and the subspace of the corresponding $d-1$ eigenvectors is spanned by the vectors $\psib = \left[\alpha_1 \vb_1^{\T} , \ldots , \alpha_d \vb_d^{\T}\right]^{\T}$, such that the scalars $\alpha_i$'s satisfy $\sum_{i=1}^d \alpha_i = 0$.
\item For each $1 \leq \ell \leq m-d$, %
  if we partition the corresponding eigenvector $\psib^{(\ell )}$ into $|\cX_i|$-dimensional subvectors $\psib_i^{(\ell )}$, such that 
\begin{align} \label{eq:Bmatrixa}
\psib^{(\ell )} = 
\left[
\begin{array}{c}
\psib^{(\ell )}_1  \\
\vdots  \\
\psib^{(\ell )}_d 
\end{array} 
\right],
\end{align}
then $\psib_i^{(\ell)}$ is orthogonal to $\vb_i$, for all $i$.
\end{enumerate}
\end{lemma}
\begin{proof}
  See Appendix \ref{sec:a:lem:1}.
\end{proof}

It will also be convenient to define the matrix $\dtmt$:
\begin{align}
  \dtmt \triangleq \dtm - d \cdot \psib^{(0)} (\psib^{(0)})^{\T}.
  \label{eq:dtmt}
\end{align}
Then from Lemma~\ref{lem:1}, the eigenvalues of $\dtmt$ are $\lambda^{(1)} \geq \cdots \geq \lambda^{(m-d)} \geq 0 = \lambda^{(m-d+1)} = \cdots = \lambda^{(m)}$, with the corresponding eigenvectors $\psib^{(1)}, \ldots , \psib^{(m-1)}, \psib^{(0)}$. 
Moreover, we define a collection of functions $f_i^{(\ell )} : \cX_i \mapsto \mathbb{R}$ as 
\begin{align} \label{eq:fpsi}
  f_i^{(\ell )}(x_i) = \frac{\psi_i^{(\ell )}(x_i)}{\sqrt{P_{X_i}(x_i)}}, \qquad \text{for all $i,\ell $},
\end{align}
where $\psib_i^{(\ell )}$ is the $i$-th subvector of $\psib^{(\ell )}$ as defined in~\eqref{eq:Bmatrixa}. Then, it follows from Lemma~\ref{lem:1} and~\eqref{eq:fpsi} that $f_i^{(\ell )}(X_i)$'s are zero-mean functions and $\sum_{i = 1}^d\E[(f^{(\ell )}_i(X_i))^2] = 1$. In addition, these functions induce an exponential family of joint distributions for $U, X^d$.

\begin{definition} \label{def:1}
  Let ${\cal H}$ be the set of functions $h : \cU \mapsto \mathbb{R}$ with zero mean and unit variance. Then, an exponential family ${\cal P}^{(\delta)}_{\exp}$ on $U, X^d$ %
  is defined as
\begin{align*}
{\cal P}^{(\delta)}_{\exp} = \Bigg \{ &\frac{1}{Z}P_U(u) 
P_{X^d} (x^d) %
  \cdot \exp \left( \sqrt{2\delta} \frac{h(u)}{\sqrt{\lambda^{(1)}}} \sum_{i = 1}^d f_i^{(1)} (x_i)
 \right) \colon h \in {\cal H} \Bigg\},
\end{align*}
where $Z$ is the normalizing factor.
\end{definition}
Note that this also defines a family of random variables $U$ embedded in $X^d$ corresponding to the collection of distributions in ${\cal P}^{(\delta)}_{\exp}$. 
It turns out that this exponential family characterizes the optimal solution of~\eqref{eq:total_correlation2} in the regime of small $\delta$, which is demonstrated as follows.

\begin{theorem} \label{thm:1}
The optimal value of~\eqref{eq:total_correlation2a} is
\begin{align} \label{eq:optimal_value1}
\max_{P_{U X^d}} \cL( X^d | U) = \delta\left( \lambda^{(1)} - 1\right) + o(\delta),
\end{align}
which is attainable by the distributions in ${\cal P}^{(\delta)}_{\exp}$. Moreover, for any distribution $P_{U X^d}$ achieving~\eqref{eq:optimal_value1}, there exists a distribution $\hat{P}_{U X^d} \in {\cal P}^{(\delta)}_{\exp}$, such that for all $(u, x^d) \in \cU \times \cX^d$,
\begin{align*}
\left| P_{U X^d} (u, x^d) - \hat{P}_{U X^d} (u, x^d) \right| = o\left(\sqrt{\delta}\right).
\end{align*}
\end{theorem}

\begin{proof}
  See Appendix \ref{sec:a:thm:1}.
\end{proof}
%
%
%
%
%
%

%

From Theorem~\ref{thm:1}, the family of random variables $U$ embedded in $X^d$ defined by ${\cal P}_{\exp}^{(\delta)}$ are the set of attributes that contain the most amount of information about the common structure shared by $X^d$. To extract such information from data, we consider the log-likelihood function to estimate $U$ from $X^d$:%
\begin{align} \label{eq:LKH}
\log \frac{P_{X^d | U = u}(x^d)}{P_{X^d}(x^d)}
= \frac{\sqrt{2\delta}h(u)}{\sqrt{\lambda^{(1)}}}  \sum_{i = 1}^d f_i^{(1)} (x_i) + o\left(\sqrt{\delta} \right).
\end{align}
Although the log-likelihood functions for different $U = u$ in the exponential family ${\cal P}^{(\delta)}_{\exp}$ may have different magnitudes due to $h(u)$, all of them are proportional to the functional representation $\sum_{i = 1}^d f_i^{(1)} (x_i)$ of the data vectors. This can be interpreted as the 1-dimensional subspace of the functional space of $X^d$ that is the most informative about the shared structure. This is similar to what linear PCA~\cite{Jolliffe} aims to achieve in the space of data, while we are searching for the optimal subspace of the general functional space. Later on we will show that our result is indeed a nonlinear generalization of linear PCA.

In addition, note that $\psib^{(1)}$ is the second largest eigenvector of $\dtm$, which maximizes $\psib^\T \dtm \psib$ over all unit vectors $\psib$ orthogonal to $\psib^{(0)}$. This implies that the functions $f_i^{(1)} (X_i)$ defined from~\eqref{eq:fpsi} form the optimal solution of the joint correlation optimization problem:
\begin{align*}
\max_{f_i \colon \cX_i \mapsto \mathbb{R}, \ i = 1, \ldots , d} \ &\E \left[ \sum_{i \neq j} f_i (X_i) f_j (X_j) \right] \\ \notag
  \text{subject to:} \quad &\E \left[ f_i (X_i) \right] =  0, \quad i = 1, \dots, d,\\
                           &\E \left[ \sum_{i = 1}^d f_i^2(X_i) \right] = 1, \quad i = 1, \dots, d.%
\end{align*}
Therefore, the functional representation $f_i^{(1)} (X_i)$ essentially searches for a 1-dimensional subspace for each functional space of $X_i$, such that the joint correlation between these subspaces is maximized. As a consequence, these subspaces and the corresponding functional representations convey much information about the common structure shared among these random variables.

\subsection{The Informative $k$-dimensional Attributes}
\label{sec:inform-k-dimens}

In addition to the largest eigenvector $\psib^{(1)}$, the rest eigenvectors of $\dtmt$ essentially lead to functional representations, which correspond to informative $k$-dimensional attributes for the common structure. To show that, we consider the optimization problem for $k$-dimensional attribute $U^k = (U_1 , \ldots , U_k)$:%
\begin{align} \label{eq:total_correlation2ak}
\max_{P_{U^k X^d }} &{\cal L} \left( X^d  | U^k \right), %
\end{align}
where ${\cal L} \left( X^d | U^k \right)$ is as defined in~\eqref{eq:total_correlation}, and the maximization is over all joint distributions $P_{U^k X^d}$ such that the constituent variables $U_i$ with ranges $\cU_i$, for $i = 1, \dots , k$, satisfy: 1) $\delta \geq I(U_1; X^d) \geq \dots \geq I(U_k; X^d)$; 2) $\min_{u_i\in\cU_i} P_{U_i} (u_i) > \gamma$ for all $i = 1, \dots, d$ and some constant $\gamma > 0$ independent of $\delta$; 3) $U_1 , \ldots , U_k$ are mutually independent variables; 4) $U_1 , \ldots , U_k$ are conditionally independent variables given $X^d$. To solve the optimization problem~\eqref{eq:total_correlation2ak}, we define the following exponential family for $k$-dimensional attributes. 
\begin{definition} \label{def:2}
  Let ${\cH}_i$ be the set of functions $h_i : \cU_i \mapsto \mathbb{R}$ with zero mean and unit variance, for $i = 1, \dots, k$. Then, an exponential family $\cP^{(\delta)}_{\exp,k}$ on $U^k, X^d$ is defined as

\begin{align*}
  {\cP}^{(\delta)}_{\exp,k}
  &= \left \{ \frac{1}{Z_k} \left[\prod_{j = 1}^k P_{U_i}(u_i)\right] P_{X^d} (x^d) %
    \cdot \exp \left( \sqrt{2\delta} \sum_{\ell = 1}^{k_0} h_\ell(u_{\ell})\sum_{j = 1}^{k_0} \frac{q_{j\ell}}{\sqrt{\lambda^{(j)}}} \sum_{i = 1}^d f_i^{(j)}(x_i) \right) \right.\\
  &\qquad\qquad \left.\colon h_\ell \in {\cH}_\ell, \Qb = [q_{ij}]_{k_0 \times k_0}, \Qb^{\T}\Qb = \Ib_{k_0}\vphantom{\sum_{j = 1}^{k}}\right\},
\end{align*}
where $f_i^{(j)}(x_i)$ is as defined in~\eqref{eq:fpsi}, $Z_k$ is the normalizing factor, $k_0 = \min\{k, k^*\}$ and
\begin{align*}
  k^* \defeq \max\left\{i\colon \lambda^{(i)} > 1\right\}.
\end{align*}
\end{definition}
Then, the exponential family ${\cal P}^{(\delta)}_{\exp,k}$ characterizes the optimal solutions of~\eqref{eq:total_correlation2ak}.
\begin{theorem} \label{thm:2}
The optimal value of~\eqref{eq:total_correlation2ak} is
\begin{align} \label{eq:optimal_value2}
\max_{P_{U^k X^d}} &{\cal L} \left( X^d | U^k \right) = \delta\left( \sum_{\ell =1}^{k_0} \lambda^{(\ell )} - k_0\right) + o(\delta),
\end{align}
which is attainable by the distributions in ${\cal P}^{(\delta)}_{\exp,k}$. Moreover, for any distribution $P_{U^k X^d}$ achieving~\eqref{eq:optimal_value2}, there exists a distribution $\hat{P}_{U^k X^d} \in {\cal P}^{(\delta)}_{\exp,k}$, such that for all $(u^k, x^d) \in \cU_1 \times \dots \times \cU_k \times \cX^d$, 
\begin{align*}
\left| P_{U^k X^d} (u^k, x^d) - \hat{P}_{U^k X^d} (u^k, x^d) \right| = o\left(\sqrt{\delta}\right).
\end{align*}
\end{theorem}
\begin{proof}  
  See Appendix \ref{sec:a:thm:2}.  
\end{proof}

Note that from Definition~\ref{def:2} and Theorem~\ref{thm:2}, when $k > k^*$, 
the optimal solution is to design $P_{U^{k^*}X^d}$ to follow the distributions in ${\cal P}^{(\delta)}_{\exp,k^*}$, and let the last $k - k^*$ attributes $U_{k^* + 1}, \dots, U_k$ be independent of $X^d$. This implies that only the top $k^*$ attributes can effectively reduce the total conditional correlation, which leads to an intrinsic criterion for designing the dimensionality $k$ of the attributes. 

Moreover from Definition~\ref{def:2}, the log-likelihood functions for the optimal attributes correspond to the functional representations $\sum_{i = 1}^d f_i^{(\ell)} (x_i)$, for $\ell = 1, \ldots , k$. This generalizes~\eqref{eq:LKH} for providing the informative $k$-dimensional representations about the common structure shared by $X_1, \ldots , X_d$. Furthermore, it is shown in Appendix \ref{sec:a:mace} that the functions $f_i^{(\ell)}$ as defined in~\eqref{eq:fpsi} form the optimal solution of the following optimization problem:
\begin{subequations}
  \begin{alignat}{2} 
    &\max_{\uf_i\colon \cX_i \mapsto \mathbb{R}^k, \ i = 1, \ldots , d}  & ~ &\E \left[ \sum_{i \neq j} \uf_i^{\T} (X_i) \uf_j (X_j) \right] \label{eq:MHGRk:1}\\ %
    &\quad\quad\text{subject to:} \ & &\E \left[ \uf_i (X_i) \right] =  \underline{0}, \  \text{for all $i$}, \\
    & & &\E \left[ \sum_{i = 1}^d \uf_i(X_i) \uf_i^{\T}(X_i) \right] = \mathbf{I}_k,
  \end{alignat}
  \label{eq:MHGRk}
\end{subequations}
where $\mathbf{I}_k$ is the $k$-dimensional identity matrix. Therefore for $i = 1, \ldots , d$, the functional representations $f_i^{(\ell)} (X_i),  \ell = 1, \ldots , k$, form the $k$-dimensional functional subspace of $X_i$, such that the joint correlation between these subspaces for different $X_i$'s is maximized.%

\begin{example}[Common Bits Patterns Extraction]
  Suppose that $b_1, \dots, b_r \in \{1, -1\}$ are mutually independent $\mathsf{Bern}(\frac{1}{2})$ bits, and each random variable $X_i = b_{\cI_i} \defeq (b_j)_{j \in \cI_i}$ %
  is a subset of these random bits, where $\cI_i \subseteq \{1, \dots, r\}$ denotes the index set. Then, our information theoretic approach essentially extracts the bit patterns that appear the most among the random variables $X^d$. To show that, we define $w(\cI)$ as the number of sets $\cI_i~(i = 1, \dots, d)$ that include $\cI$, i.e., 
\begin{align}
  w(\cI)
  \defeq  \sum_{i = 1}^d \1_{\{\cI \subset \cI_i\}},
  \label{eq:def:w}
\end{align}
where $\1_{\cdot}$ is the indicator function. In addition, we denote $\varnothing = \cJ_0, \dots, \cJ_{2^r - 1}$ as the $2^r$ subsets of $\{1, \dots, r\}$ with the decreasing order
 $ d = w(\cJ_0) \geq w(\cJ_1) \geq \dots \geq w(\cJ_{2^{r} - 1})$. 
Then, it is shown in Appendix \ref{sec:a:prop:bits} that the eigenvalues for the corresponding matrix $\dtm$ are 
  \begin{align}
    \lambda^{(\ell)} = w(\cJ_{\ell}), \quad\ell = 0, \dots, m - 1.
    \label{eq:bits:lambda}
  \end{align}
where $m = \sum_{i = 1}^d 2^{|\cI_i |}$ is the dimensionality of $\dtm$. Therefore, the eigenvalue $\lambda^{(\ell)}$ of the matrix $\dtm$ essentially counts the number of times the corresponding bits pattern $b_{\cJ_{\ell}}$ appears among the random variables $X^d$, and the largest eigenvalue indicates the most appeared bits pattern. Moreover for $\lambda^{(\ell)} > 0$, the corresponding functions $f^{(\ell)}_i(X_i)$ ($i = 1, \dots, d$) as defined in \eqref{eq:fpsi} are
  \begin{align}
    f^{(\ell)}_i(X_i) =
    \begin{cases}
      \displaystyle\frac{1}{\sqrt{w(\cJ_{\ell})}}\prod_{j \in \cJ_{\ell}} b_j,&\text{if}~ \cJ_{\ell} \subset \cI_i,\\
      0,&\text{otherwise}.
    \end{cases}
          \label{eq:f:bits}
  \end{align}
Thus, the $\ell$-th optimal functional representation of $X^d$ (cf. Section \ref{sec:inform-k-dimens}) is
\begin{align*}
  \sum_{i = 1}^d f_i^{(\ell)}(X_i) = \sqrt{w(\cJ_{\ell})} \prod_{j \in \cJ_{\ell}} b_j,
\end{align*}
which depends only on the bits indexed by $\cJ_{\ell}$. For instance, if $r=d=3$, and $X_1 = \{b_1, b_2\}, X_2 = \{b_2, b_3\}$, and $X_3 = \{b_1, b_3\}$, then for all subsets of $\{1, 2, 3\}$, the values for the function $w(\cdot)$ as defined in \eqref{eq:def:w} are
  \begin{gather*}
    w(\varnothing) = 3, \quad%
    w(\{1\}) = w(\{2\}) = w(\{3\}) = 2,\\
    w(\{1, 2\}) = w(\{2, 3\}) = w(\{3, 1\}) = 1,\quad%
    w(\{1, 2, 3\}) = 0.
  \end{gather*}
  Therefore, the corresponding eigenvalues of $\dtm$ are
  \begin{gather*}
    \lambda^{(0)} = 3,\quad%
    \lambda^{(1)} = \lambda^{(2)} = \lambda^{(3)} = 2,\quad%
    \lambda^{(4)} = \lambda^{(5)} = \lambda^{(6)} = 1,\quad%
    \lambda^{(7)} = 0.
  \end{gather*}
  Moreover, the corresponding $f^{(\ell)}_i(X_i)$'s satisfy
  \begin{align*}
    \sum_{i = 1}^3 f_i^{(\ell)}(X_i) =  \sqrt{2} b_\ell, \quad \ell = 1, 2, 3,
  \end{align*}
  and
  \begin{align*}
    \sum_{i = 1}^3 f_i^{(\ell)}(X_i) =
    \begin{cases}
      b_1b_2,&\ell = 4,\\
      b_2b_3,&\ell = 5,\\
      b_3b_1,&\ell = 6.
    \end{cases}
  \end{align*}

\end{example}

\section{The Algorithm to Compute the Functional Representation from Data} \label{sec:3}

While our information theoretic approach provides a guidance for searching informative functional representations, it remains to derive the algorithm to compute these functions from observed data vectors. Intuitively, one can first estimate the empirical distribution between $X_1 , \ldots , X_d$ from the data samples, and then construct the matrix $\dtm$ to solve the eigen-decomposition. However, this is often not feasible in practice due to: (1) there may not be enough number of samples to estimate the joint distribution accurately, (2) the dimensionality of $\dtm$ may be extremely high especially for big data applications, so that the singular value decomposition (SVD) can not be computed directly. 

\subsection{The Multivariate Alternating Conditional Expectation (MACE) Algorithm}

Alternatively, it is well-known that eigenvectors of a matrix can be efficiently computed by the power method~\cite{golub2013matrix}. The power method iteratively multiplies the matrix to an initial vector, and if all the eigenvalues are nonnegative, it converges to the eigenvector with respect to the largest eigenvalue with an exponential convergence rate. To apply the power method for computing the second largest eigenvector of $\dtm$, we choose an initial vector $\psib = 
\begin{bmatrix}
\psib_1^{\T}& \cdots& \psib_d^{\T}
\end{bmatrix}
^\T$, such that $\psib_i$ is orthogonal to $\vb_i$, for all $i$. This forces $\psib$ to be orthogonal to $\psib^{(0)}$, and since $\dtm$ is positive semidefinite, the power iteration will converge to the second largest eigenvector if $\psib$ is not orthogonal to $\psib^{(1)}$. Then, the algorithm iteratively computes the matrix multiplication $\psib \leftarrow \dtm \psib$, or equivalently 
\begin{align} \label{eq:ace_i}
\psib_i \leftarrow \psib_i + \sum_{j \neq i} \dtm_{ij}\psib_j, 
\end{align}
for all $i$. Note that if we write $f_i(x_i) = \psi_i (x_i)/\sqrt{P_{X_i}(x_i)}$, then as shown in~\cite{SL17}, the step~\eqref{eq:ace_i} is equivalent to a conditional expectation operation on functions:
\begin{align} \label{eq:ace_ffff}
f_i (X_i) \leftarrow f_i (X_i) +  \E \left[ \left. \sum_{j \neq i} f_j (X_j) \right| X_i \right], 
\end{align}
Therefore, the power method can be transferred to an algorithm based on the alternating conditional expectation (ACE)~\cite{Breiman} algorithm as shown in Algorithm~\ref{alg:1}, which computes the optimal functional representation derived in Section~\ref{sec:2}. Note that the choice of $\psib_i$ to be orthogonal to $\vb_i$ is transferred to the zero-mean choice of functions in the initialization step of the algorithm. 
\begin{algorithm} [t]
\caption{The Multivariate ACE (MACE) Algorithm}
$\mathbf{Require:}$ The data samples $\ux^{(\ell)} = (x_1^{(\ell)} , \ldots , x_d^{(\ell)})$, $\ell = 1, \dots, n$ of variables $X_1, \ldots , X_d$
\begin{itemize}
\item [1.] Initialization: randomly pick zero-mean functions $\vec{f} = (f_1, \ldots , f_d)$.
\end{itemize}
\text{}\text{} $\mathbf{repeat:}$
\begin{itemize}
\item [2.] The alternating conditional operation: $f_i (X_i) \leftarrow  f_i (X_i) + \E \left[ \left. \sum_{j \neq i} f_j (X_j) \right| X_i \right].$
\item [3.] The normalization step: $f_i (X_i) \leftarrow f_i (X_i) / \sqrt{ \E \left[ \sum_{i = 1}^d f_i^2(X_i)  \right]}.$ %
\end{itemize}
\text{}\text{} $\mathbf{until:}$ \ The $\E \left[ \sum_{i \neq j} f_i(X_i) f_j(X_j) \right]$ stops increasing.
\label{alg:1}
\end{algorithm}

\subsection{Finding $k$ Functional Representations from Eigen-decomposition}

The Algorithm~\ref{alg:1} can be further extended to compute the top $k$ eigenvectors $\psib^{(1)}, \ldots , \psib^{(k)}$, and the corresponding functional representations. To design the algorithm for computing these functions, we denote the $\ell$-th functional representation as $\vec{f}^{(\ell)} = ( f_1^{(\ell)}, \ldots , f_d^{(\ell)} )$, where $f_i^{(\ell)}$ is as defined in~\eqref{eq:fpsi}. Then, since $\psib^{(k)}$ is orthogonal to $\psib^{(\ell)}$, for $ \ell \leq k-1$, the $k$-th functional representation $\vec{f}^{(k)}$ can be computed by the power method similar to the first functional representation $\vec{f}^{(1)}$, but with extra orthogonality constraints
\begin{align*}
\left\langle \vec{f}^{(\ell)}, \vec{f}^{(k)} \right\rangle \triangleq  \sum_{i = 1}^d \E \left[ f_i^{(\ell)}(X_i) f_i^{(k)}(X_i) \right] = 0, \ \text{for $\ell \leq k-1$}
\end{align*}
to maintain the orthogonality to the first $k-1$ functional representations. Therefore, $\vec{f}^{(k)}$ can be computed by the power method as in Algorithm~\ref{alg:1} with the extra step of Gram-Schmidt procedure to guarantee the orthogonality, which is illustrated in Algorithm~\ref{alg:2}. Note that the computation complexities of Algorithm~\ref{alg:1} and Algorithm~\ref{alg:2} are both linear to the size of the dataset, which is often much more efficient than the singular value decomposition of the matrix $\dtm$.

\begin{algorithm}[t]
  \caption{The Computation of $\vec{f}^{(k)}$}
  $\mathbf{Require:}$ The data samples $\ux^{(i)} = (x_1^{(i)} , \ldots , x_d^{(i)})$, $i = 1, \dots, n$ of variables $X_1, \ldots , X_d$, and the previously computed functions $\vec{f}^{(1)}, \ldots , \vec{f}^{(k-1)}$.
  \begin{itemize}
  \item [1.] Initialization: randomly pick zero-mean functions $\vec{f}^{(k)} = ( f_1^{(k)}, \ldots , f_d^{(k)} )$.
  \end{itemize}
  \text{}\text{} $\mathbf{repeat:}$
  \begin{itemize}
  \item [2.] Run step 2 and 3 of Algorithm~\ref{alg:1} for $\vec{f}^{(k)}$.
  \item [3.] The Gram-Schmidt procedure: $\vec{f}^{(k)} \leftarrow \vec{f}^{(k)} - \sum_{\ell = 1}^{k-1} \langle \vec{f}^{(\ell)}, \vec{f}^{(k)} \rangle \cdot \vec{f}^{(\ell)}$
  \end{itemize}
  \text{}\text{} $\mathbf{until:}$ \ The $\E \left[ \sum_{i \neq j} f^{(k)}_i(X_i) f^{(k)}_j(X_j) \right]$ stops increasing.
  \label{alg:2}
\end{algorithm}

\subsection{Generating Informative Functional Representations for High-Dimensional Data} \label{sec:low-rank}

While the Algorithm~\ref{alg:2} generally requires less training samples than estimating the joint distribution and the matrix $\dtm$, in order to obtain an acceptable estimation for the conditional expectation step~\eqref{eq:ace_ffff}, it is still necessary to acquire training samples in the size comparable to the cardinality of the random variable $X_i$. This is often difficult for high-dimensional or continuous random variables in practice. In such cases, we propose a neural network based approach to generate the informative functional representations by deep neural networks. The key idea is to note that by Eckart-Young-Mirsky theorem \cite{eckart1936approximation}, the top $k \leq m$ eigenvectors of $\dtmt$ can be computed from the low-rank approximation problem:
\begin{align} \label{eq:EYM}
  \Psib^* = \min_{\Psib \in \mathbb{R}^{m \times k}} \bfrob{\dtmt - \Psib\Psib^{\T}}^2
\end{align}
where the columns of $\Psib^*$ are the top $k$ eigenvectors of $\dtmt$. The unconstrained optimization problem~\eqref{eq:EYM} leads to a training loss for generating informative functions by neural networks.
\begin{prop} \label{prop:MH-score}
  Let $\Psib_i$ be $|\cX_i| \times k$ matrices, for $i = 1, \ldots , d$, such that $\Psib = \left[ \Psib_1^{\T} \cdots \Psib_d^{\T}  \right]^\T$, and define $k$-dimensional functions $\uf_i\colon \cX_i \mapsto \mathbb{R}^k$, $i = 1, \dots, d$, as $\uf_i (x_i) = \Psib_i^{\T} (x_i) / \sqrt{P_{X_i}(x_i)}$, where $\Psib_i (x_i)$ denotes the $x_i$-th row of the matrix $\Psib_i$. %
  Then, it follows that
\begin{align} \label{eq:MH-loss}
  \bfrob{\dtmt - \Psib\Psib^{\T}}^2 = \bfrob{\dtmt}^2 - 2H\left(\uf_1(X_1), \dots, \uf_d(X_d)\right),%
\end{align}
where
\begin{align*}
H\left(\uf_1(X_1), \dots, \uf_d(X_d)\right) \defeq \sum_{i = 1}^d \sum_{j = 1}^d H\left(\uf_i(X_i), \uf_j(X_j)\right),
\end{align*}
and $H\left(\uf_i(X_i), \uf_j(X_j) \right)$ is defined as, for all $i,j$,
\begin{align*}
  H\left(\uf_i(X_i), \uf_j(X_j)\right)%
&\defeq \E\left[\uf^{\T}_i(X_i)\uf_j(X_j)\right] - \left(\E\left[\uf_i(X_i)\right]\right)^{\T}\E\left[\uf_j(X_j)\right]\\
&\qquad  - \frac12 \trop{\E\left[\uf_i(X_i)\uf_i^\T(X_i)\right]\E\left[\uf_j(X_j)\uf_j^\T(X_j)\right]}, 
\end{align*}
where $\trop{\cdot}$ denotes the trace of its matrix argument.
\end{prop}
\begin{proof}
See Appendix \ref{sec:a:MH}.
\end{proof} 

Note that $H\left(\uf_i(X_i), \uf_j(X_j)\right)$ coincides with the \emph{H-score}~\cite{huang2019information} when the means of the functions are zero, hence we term $H\left(\uf_1(X_1), \dots, \uf_d(X_d)\right)$ the \emph{multivariate H-score} (MH-score). Then from~\eqref{eq:MH-loss}, the optimization problem~\eqref{eq:EYM} is equivalent to the functional optimization problem
 \begin{align} \label{eq:MHLOSS}
\max_{f_i\colon \cX_i \mapsto \mathbb{R}^k, \ i = 1, \ldots, d} H\left(\uf_1(X_1), \dots, \uf_d(X_d)\right),
\end{align}
for solving the informative functional representations for the common structure. The optimization problem~\eqref{eq:MHLOSS} leads to a neural network training strategy. Specifically, given the training samples of $X_1, \dots, X_d$, we design $d$ neural networks, where the $i$-th neural network $\mathtt{NN}_i$ takes $X_i$ as the input and generates the representations $\uf_i(X_i)$. Then, the weights of these neural networks are trained to minimize the negative MH-score as the loss function. Finally, the informative functional representations are generated by the trained $d$ neural networks that attempts to optimize~\eqref{eq:MHLOSS}, as illustrated in Fig. \ref{fig:mh}.

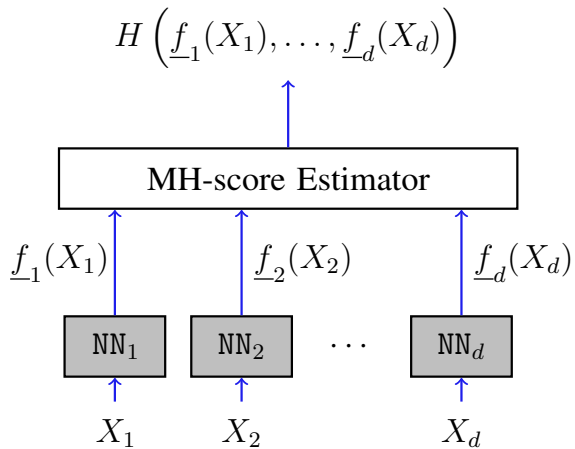
\begin{figure}[!t]
    \centering
    \resizebox{.5\textwidth} {!}{\tikzset{%
  nn/.style    = {draw = black, fill = gray!50!white, thick, rectangle, minimum height = 1.8em, minimum width = 3em, text width = 1.8em, align = center},
  nnfake/.style    = {nn, draw = none, fill = none},
  hscore/.style    = {draw = black, thick, rectangle, minimum height = 1.8em,  minimum width = 3em, text width = 13em, align=center},
  input/.style    = {coordinate}, 
  output/.style   = {coordinate}, 
  *|/.style={
        to path={
            (perpendicular cs: horizontal line through={(\tikztostart)},
                                 vertical line through={(\tikztotarget)})
            -- (\tikztotarget) \tikztonodes
        }
    }
}
\def\nndist{.3cm}

\begin{tikzpicture}[auto, thick, node distance=1.8cm, ->,draw=white!60!black!20!blue]
  \draw
  node [nn] (nn1) {$\mathtt{NN}_1$}
  node [nn, right = \nndist of nn1] (nn2) {$\mathtt{NN}_2$}
  node [draw = none, right = \nndist of nn2] (cdots) {$\cdots$}
  node [nn, right = \nndist of cdots] (nn3) {$\mathtt{NN}_{d}$};
  \draw
  node [nnfake, below = \nndist of nn1] (i1) {$X_1$}
  node [nnfake, below = \nndist of nn2] (i2) {$X_2$}
  node [nnfake, below = \nndist of nn3] (i3) {$X_{d}$};

  \draw[->](i1) -- node[midway, input](mid1){} node[left] {}(nn1);
  \draw[->](i2) -- node[midway, input](mid2){} node[right] {}(nn2);
  \draw[->](i3) -- node[midway, input](mid3){} node[right] {}(nn3);

  \node (mid) at ($(mid1)!0.5!(mid3)$) {};
  \draw
  node [hscore, above = 5em of mid] (h-score) {MH-score Estimator}
  node [nnfake, text width = 12em, above = 2em of h-score, name = out] {$H\left(\uf_1(X_1), \dots, \uf_{d}(X_{d})\right)$};

  \draw[<-, *|](h-score.south) to node[left, xshift = .2em]{$\uf_1(X_1)$} (nn1.north);
  \draw[<-, *|](h-score.south) to node[right, xshift = 0em]{$\uf_2(X_2)$} (nn2.north);
  \draw[<-, *|](h-score.south) to node[right, xshift = 0em]{$\uf_{d}(X_{d})$} (nn3.north);
  
  \draw[->](h-score) -- (out);
\end{tikzpicture}

}
    \caption{The network architecture to estimate optimal functional representations, where each $\mathtt{NN}_i$ is a neural network to extract feature $\uf_i(X_i)$ from the $i$-th input $X_i$.}%
    \label{fig:mh}
\end{figure}

\if\secsc1

\section{The Asymptotic Sample Complexity of MACE Algorithm}
In the MACE algorithm, the conditional expectations (step 2 and 3 of Algorithm \ref{alg:1}) are estimated from the data samples, which are assumed to be i.i.d. generated from some unknown joint distribution. Due to the i.i.d. sampling process, there exists a deviation between the empirical and true distributions, which results in the deviations between the empirical and true conditional expectations in the MACE algorithm. This leads to the sample complexity problem for computing the functional representations in the MACE algorithm, which is studied in this section.

Specifically, for the functional representations $\fh_1^{(\ell)}, \ldots , \fh_d^{(\ell)}$, $\ell = 1, \dots, k,$ computed by the MACE algorithm from $n$ i.i.d. data samples of $X^d$ with the empirical distribution $\Ph_{X^d}$, we define
\begin{align*}
  \Psibh =
  \begin{bmatrix}
    \psibh^{(1)} & \cdots & \psibh^{(k)}
  \end{bmatrix}
\end{align*}
where $\psibh^{(\ell)}$ is defined as [cf. \eqref{eq:Bmatrixa}]
\begin{align} \label{eq:psibh}
\psibh^{(\ell )} = 
\left[
\begin{array}{c}
\psibh^{(\ell )}_1  \\
\vdots  \\
\psibh^{(\ell )}_d 
\end{array} 
\right]
\end{align}
with $\psih^{(\ell )}_i(x_i) = \sqrt{\Ph_{X_i}(x_i)} \fh_i^{(\ell)}$ for $x_i \in \cX_i$ and $i = 1, \dots, d$. To measure the deviation between $\Psibh$ and the top $k$ eigenvectors of $\dtmt$, denoted by
\begin{align*}
  \Psib = \begin{bmatrix}
    \psib^{(1)} & \cdots & \psib^{(k)}
  \end{bmatrix},
\end{align*}
we apply 
\begin{align}
  \trop{\Psib^{\T} \dtmt \Psib} - \trop{\Psibh^{\T} \dtmt \Psibh} %
  \label{eq:deviation}
\end{align}
as the performance metric to characterize the sample complexity of the MACE algorithm, which corresponds to the loss of the joint correlation (cf. Appendix \ref{sec:a:mace}).

Note that the eigenvectors $\Psib$ satisfies
\begin{align*}
    \trop{\Psib^{\T} \dtmt \Psib}%
    = \sum_{i = 1}^k \lambda^{(i)},
\end{align*}
and $  \trop{\Psib^{\T} \dtmt \Psib} - \trop{\Psibh^{\T} \dtmt \Psibh} \geq 0$, for all $m \times k$ matrices satisfying $\Psibh^{\T}\Psibh = \Ib_{k}$, where the equality holds when $\Psibh = \Psib$. Our goal in this section is to characterize the (asymptotic) error exponent
\begin{align}
  -\lim_{\eps \rightarrow 0} \frac{1}{\eps^2} \lim_{n \rightarrow \infty} \frac{1}{n} \log \mathbb{P}_n \left\{   \trop{\Psib^{\T} \dtmt \Psib} - \trop{\Psibh^{\T} \dtmt \Psibh}%
  > \eps^2  \right\},
  \label{eq:exponent1}
\end{align}
where the probability is measured over the i.i.d. sampling process from $P_{X^d}$. In particular, the first limit indicates the asymptotic regime of $n$ we are interested in, and the second limit of $\eps$ is naturally from that in the asymptotic regime of $n$, the difference between the estimation
 $\trop{\Psibh^{\T} \dtmt \Psibh}$  %
and the optimal value $\trop{\Psib^{\T} \dtmt \Psib}$ %
is small. Moreover, we assume that $\lambda^{(k)} > \lambda^{(k + 1)}$, which is often a reasonable assumption in practical problems, and will facilitate our derivation.

To investigate the error exponent \eqref{eq:exponent1}, note that the MACE algorithm output $\Psibh$, computed from the i.i.d. samples, corresponds to the top $k$ eigenvectors of a matrix $\dtmh$, defined from the empirical distribution $\Ph_{X^d}$ as
\begin{align} \label{eq:dtmh}
  \dtmh = 
  \left[
  \begin{array}{cccc}
    \dtmh_{11} & \dtmh_{12} & \cdots & \dtmh_{1d} \\
    \dtmh_{21} & \dtmh_{22} & \cdots & \dtm_{2d} \\
    \vdots & \vdots & \ddots & \vdots \\
    \dtmh_{d1} & \dtmh_{d2}  & \cdots & \dtmh_{dd}
  \end{array} 
                                        \right],
\end{align}
where $\dtmh_{ij}$ is an $|\cX_i| \times |\cX_j|$ matrix with entries
\begin{align}
  \Bh_{ij}(x_i; x_j) = \frac{\Ph_{X_i X_j}(x_i, x_j) - \Ph_{X_i}(x_i)\Ph_{X_j}(x_j)}{\sqrt{\Ph_{X_i}(x_i)}\sqrt{\Ph_{X_j}(x_j)}}
\end{align}
for $i \neq j$, and where $\dtmh_{ii}$ is an $|\cX_i| \times |\cX_i|$ matrix with entries
\begin{align}
  \Bh_{ii}(x_i; x'_i) = \delta_{x_ix_i'} - \sqrt{\Ph_{X_i}(x_i)\Ph_{X_i}(x_i')},\label{eq:Bh:ii}
\end{align}
and where $\delta_{x_ix_i'}$ denotes the Kronecker delta. 

In what follows, we use an $(|\cX_1|\cdot|\cX_2|\cdots|\cX_d|)$-dimensional vector $\xib_{X^d}$ to denote the difference between the empirical and true distributions of $X^d$, defined as
\begin{align}
  \xi_{X^d}(x^d) \defeq
    \begin{cases}
      \displaystyle
      \frac{\hat{P}_{X^d}(x^d) - P_{X^d}(x^d)}{\eps\sqrt{P_{X^d}(x^d)}},&\text{if}~P_{X^d}(x^d) > 0,\\
      0,&\text{if}~P_{X^d}(x^d) = 0,
    \end{cases}
    \label{eq:xi:x^d}
\end{align}
where $\eps$ is some small quantity, since we consider the regime of large $n$. %
Note that since $\Ph_{X^d}$ and $P_{X^d}$ are probability distributions, the vector $\xib_{X^d}$ satisfies
\begin{align} \label{eq:cond_ortho}
\sum_{x^d }\sqrt{P_{X^d}(x^d)} \xi_{X^d} (x^d) = 0.
\end{align}

Similarly, for all $i \neq j$, the difference between the empirical true distributions of $X_i, X_j$ can be characterized by
\begin{align}
  &\xi_{X_iX_j}(x_i, x_j)\notag\\
  &\defeq 
    \begin{cases}
      \displaystyle
      \frac{\Ph_{X_iX_j}(x_i, x_j) - P_{X_iX_j}(x_i, x_j)}{\eps\sqrt{P_{X_iX_j}(x_i, x_j)}},&\text{if}~P_{X_iX_j}(x_i, x_j) > 0,\\
      0,&\text{if}~P_{X_iX_j}(x_i, x_j) = 0,
    \end{cases}
  \label{eq:xi:xixj}
\end{align}
for all $x_i \in \cX_i$ and $x_j \in \cX_j$.

In addition, we can represent the difference between $\Ph_{X_i}$ and $P_{X_i}$ as
\begin{align}  
  \hat{P}_{X_i}(x_i) - P_{X_i}(x_i) = \eps\sqrt{P_{X_i}(x_i)}\xi_{X_i}(x_i),
  \label{eq:xi}
\end{align}
where
\begin{align}
  \xi_{X_i}(x_i) = \sum_{x_j \in \cX_j} \sqrt{P_{X_j|X_i}(x_j|x_i)} \xi_{X_iX_j}(x_i, x_j), \quad i\neq j.
  \label{eq:xi_ij:xi_i}
\end{align}

Then, it follows from \eqref{eq:dtmh}--\eqref{eq:Bh:ii} and \eqref{eq:xi:xixj}--\eqref{eq:xi_ij:xi_i} that
\begin{align}
  \dtmh = \dtmt + \eps \Xib + o(\eps),
\end{align}
where
\begin{align}
  \Xib = 
  \left[
  \begin{array}{cccc}
    \Xib_{11} & \Xib_{12} & \cdots & \Xib_{1d} \\
    \Xib_{21} & \Xib_{22} & \cdots & \Xib_{2d} \\
    \vdots & \vdots & \ddots & \vdots \\
    \Xib_{d1} & \Xib_{d2}  & \cdots & \Xib_{dd}
  \end{array}\right],
                                      \label{eq:Xi}
\end{align}
and where $\Xib_{ij}$ is an $|\cX_i| \times |\cX_j|$ matrix with entries
\begin{align}
  &\Xi_{ii}(x_i; x_i')\notag\\
  & = - \frac12 \sqrt{P_{X_i}(x_i)P_{X_i}(x_i')}\left[\frac{\xi_{X_i}(x_i)}{\sqrt{P_{X_i}(x_i)}} + \frac{\xi_{X_i}(x_i')}{\sqrt{P_{X_i}(x_i')}}\right]
  \label{eq:Xi:ii}
\end{align}
and
\begin{align}
  &\Xi_{ij}(x_i; x_j)\notag\\
  &\defeq \frac{\sqrt{P_{X_iX_j}(x_i,x_j)}}{\sqrt{P_{X_i}(x_i)P_{X_j}(x_j)}} \xi_{X_iX_j}(x_i, x_j)\notag\\
  &\quad-  \frac{P_{X_iX_j}(x_i, x_j) + P_{X_i}(x_i)P_{X_j}(x_j)}{2\sqrt{P_{X_i}(x_i)P_{X_j}(x_j)}}\notag\\
  &\qquad\cdot\Biggl[\frac{1}{P_{X_i}(x_i)}  \sum_{x_j' \in \cX_j} \sqrt{P_{X_iX_j}(x_i, x_j')} \xi_{X_iX_j}(x_i, x_j') \notag\\
  &\qquad\quad  +\frac{1}{P_{X_j}(x_j)} \sum_{x_i' \in \cX_i}\sqrt{P_{X_iX_j}(x_i', x_j)}\xi_{X_iX_j}(x_i', x_j) \Biggr] \label{eq:Xi:ij}
\end{align}
for $i \neq j$.

Therefore, the performance metric \eqref{eq:deviation} can be characterized by $\Xib$ in the small $\eps$ regime, where the following lemma from~\cite{huang2019sample} will be useful.

\begin{lemma}[{\cite[Lemma 1]{huang2019sample}}]
  \label{lem:eig:k}
  Suppose that $\Ab \in \mathbb{R}^{m \times m}$ is a symmetric matrix with eigenvalues $\lambda_1 \geq \dots \geq \lambda_{k} > \lambda_{k + 1} \geq \dots \geq \lambda_m$ and the corresponding eigenvectors $\ub_1, \dots, \ub_m$, and let
\begin{align} \label{eq:vk}
  \Ub_k \triangleq
  \begin{bmatrix}
    \ub_1& \cdots & \ub_k
  \end{bmatrix}
  \in \mathbb{R}^{m \times k}
\end{align}
denote the matrix formed by the top $k$ eigenvectors of $\Ab$. Moreover, suppose that $\Ab(\eps)$ is an analytic function of $\eps$ with $\Ab(0) = \Ab$, and the Taylor series expansion
\begin{align*}
  \Ab(\eps) = \Ab + \eps \Ab' + o(\eps),
\end{align*}
where $\Ab' = \Ab'(0)$  is the first-order derivative of $\Ab(\eps)$ with respect to $\eps$ at $\eps = 0$. Let $\Ub_k(\eps) \in \mathbb{R}^{m \times k}$ be the matrix formed by the top $k$ eigenvectors of $\Ab(\eps)$ defined similarly to \eqref{eq:vk}, then we have
  \begin{align*}
    &\trop{ \Ub_k^{\T}(\eps) \Ab \Ub_k(\eps)} \\
    & = \trop{ \Ub_k^{ \T} \Ab \Ub_k} - \eps^2\sum_{i=1}^k\sum_{j = k + 1}^m\frac{\bigl(\ub_i^{\T}\Ab'\ub_j\bigr)^2}{\lambda_i - \lambda_j} + o(\eps^2).
  \end{align*}
\end{lemma}

Now, let us define $\cS_1(\eps)$ as the set of $\Ph_{X^d}$ such that $\trop{\Psib^{\T} \dtmt \Psib} - \trop{\Psibh^{\T} \dtmt \Psibh} > \eps^2$, 
i.e.,
\begin{align} \label{eq:s1_sup}
\cS_1(\eps) \triangleq \left\{ \Ph_{X^d}\colon \trop{\Psib^{\T} \dtmt \Psib} - \trop{\Psibh^{\T} \dtmt \Psibh} > \eps^2  \right\}.
\end{align}
Then, it follows from Sanov's theorem \cite{cover2012elements} that
\begin{align} 
&\mathbb{P}_n \left\{ \trop{\Psib^{\T} \dtmt \Psib} - \trop{\Psibh^{\T} \dtmt \Psibh} > \eps^2  \right\} \notag\\
 & \doteq \exp \left\{ -n \min_{\Ph_{X^d} \in \cS_1(\eps)} D\big(\Ph_{X^d}\big\| P_{X^d}\bigr) \right\}, \label{eq:scnn}
\end{align} 
where the ``$\doteq$'' is the conventional dot-equal notation. 

Moreover, we define $\cS_2(\eps)$ as the set of $\Ph_{X^d}$ such that the corresponding $\xib_{X^d}$ from~\eqref{eq:xi:x^d} satisfies
\begin{align}\label{eq:condition1}
  \sum_{i = 1}^k\sum_{j = k + 1}^{m} \frac{\left[\left(\psib^{(i)}\right)^{\T}
  \dtmdmat \psib^{(j)}\right]^2}{\lambda^{(i)} - \lambda^{(j)}} \geq 1
\end{align}
where $m = \sum_{i = 1}^d |\cX_i|$. Then, it follows from Lemma~\ref{lem:eig:k} that 
\begin{align}
&\lim_{\eps \rightarrow 0} \frac{1}{\eps^2} \min_{\Ph_{X^d} \in \cS_1(\eps)} D(\Ph_{X^d}\| P_{X^d}) \notag\\
&= \lim_{\eps \rightarrow 0} \frac{1}{\eps^2} \min_{\Ph_{X^d} \in \cS_2(\eps)} D(\Ph_{X^d}\| P_{X^d}).\label{eq:s1:s2}
\end{align}
In addition, from the second order Taylor's expansion of the K-L divergence we have
\begin{align}
  D\bigl(\Ph_{X^d} \big\| P_{X^d}\bigr) = \frac{\eps^2}{2} \|\xib_{X^d}\|^2 + o(\eps^2).
  \label{eq:local:kl}
\end{align}

Therefore, from \eqref{eq:scnn} and \eqref{eq:s1:s2}--\eqref{eq:local:kl}, the characterization of the sample complexity \eqref{eq:exponent} can be reduced to the minimization of $\|\xib_{X^d}\|^2$ subjected to \eqref{eq:cond_ortho} and \eqref{eq:condition1}. Since the left-hand side of ~\eqref{eq:condition1} is a quadratic function of $\xib_{X^d}$, this is equivalent to the optimization problem:
\begin{align}
&\quad\max\quad ~ \sum_{i = 1}^k\sum_{j = k + 1}^{m} \frac{\left[\left(\psib^{(i)}\right)^{\T}
 \dtmdmat \psib^{(j)}\right]^2}{\lambda^{(i)} - \lambda^{(j)}}\notag\\
&\text{subject to:} ~ \|\xib_{X^d}\|^2 \leq 1, \quad \sum_{x^d} \sqrt{P_{X^d}(x^d)}\xi_{X^d}(x^d) = 0.
  \label{eq:opt}
\end{align}
Note that for the optimal solution $\xib_{X^d}^*$ of~\eqref{eq:opt}, and for $\hat{x}^d \in \cX^d$ such that $P_{X^d}(\hat{x}^d) = 0$, we must have $\xi_{X^d}^* (\hat{x}^d) = 0$, otherwise we can set $\xi_{X^d}^*(\hat{x}^d) = 0$ and rescale $\xib_{X^d}^*$ to $\|\xib_{X^d}^*\| = 1$, which increases the objective function of \eqref{eq:opt}. Therefore, the optimal solution of \eqref{eq:opt} also satisfies \eqref{eq:xi:x^d}. Then, the analytical expression sample complexity~\eqref{eq:exponent} can be established by solving the optimization problem \eqref{eq:opt}.

To delineate the results, for each pair $(i, j)$, let $\Lb_{ij}$ be the $(|\cX_i| \cdot |\cX_j|) \times (|\cX_i| \cdot |\cX_j|)$ matrix whose entry at the $[(x'_j - 1)|\cX_i| + x_i]$-th row and $[(\hat{x}_j' - 1)|\cX_i| + \hat{x}_i]$-th column\footnote{In order to identify the entries of the matrix $\Lb_{ij}$, we may simply take the alphabets $\cX_i = \{1 , 2 , \ldots , |\cX_i| \}$ for all $i = 1, \dots, d$, which corresponds to some given alphabet orders of random variables $X_1, \dots, X_d$.} is defined as%
\begin{align} \notag
  &\sqrt{\frac{P_{X_iX_j}(\hat{x}_i, \hat{x}_j')}{P_{X_i}(x_i)P_{X_j}(x'_j)}} \biggl[\delta_{x_i\hat{x}_i}\delta_{x_j\hat{x}_j'} - \frac{1}{2}\left(\frac{\delta_{x_i\hat{x}_i}}{P_{X_i}(x_i)} + \frac{\delta_{x'_j\hat{x}_j'}}{P_{X_j}(x_j')}\right) \\  
   &\hspace{9em}
    \cdot\left[P_{X_iX_j}(x_i,x'_j)+P_{X_i}(x_i)P_{X_j}(x'_j)\right]\biggr].\label{eq:Lb:ij}
\end{align}
In addition, we define $\dtm_{ij;st}$ as an $(|\cX_i|\cdot|\cX_j|) \times (|\cX_s|\cdot|\cX_t|)$ matrix with entries $B_{ij;st}((x_i, x'_j); (\hat{x}_{s}, \hat{x}'_{t})) = 0$ if $P_{X_iX_j}(x_i, x'_j)P_{X_sX_t}(\hat{x}_s, \hat{x}'_t) = 0$ and
\begin{align}
  B_{ij;st}\left((x_i, x'_j); (\hat{x}_{s}, \hat{x}'_{t})\right) = \frac{P_{X_iX_jX_sX_t}(x_i, x'_j, \hat{x}_s, \hat{x}'_t)}{\sqrt{P_{X_iX_j}(x_i, x'_j)}\sqrt{P_{X_sX_t}(\hat{x}_s, \hat{x}'_t)}}
  \label{eq:dtm:4:entry}
\end{align}
if $P_{X_iX_j}(x_i, x'_j)P_{X_sX_t}(\hat{x}_s, \hat{x}'_t) > 0$. Finally, we define the $m^2 \times m^2$ matrix 
\begin{align}
  \Jb \defeq
  \begin{bmatrix}
    \Jb_{11;11}& \Jb_{11;21}& \cdots &\Jb_{11;dd}\\
    \Jb_{21;11}& \Jb_{21;21}& \cdots &\Jb_{21;dd}\\
    \vdots& \vdots &  \ddots & \vdots\\
    \Jb_{dd;11}  & \Jb_{dd;21}& \cdots& \Jb_{dd;dd}
  \end{bmatrix},
  \label{eq:Jb:def}
\end{align}
where the submatrix $\Jb_{ij;st}$ is defined as $\Jb_{ij;st} \defeq \Lb_{ij}\dtm_{ij;st}\Lb_{st}^{\T}$, for all $i, j, s, t \in \{1, \dots, d\}$.

Now, the sample complexity can be summarized as the following theorem.
\begin{theorem} \label{thm:exponent}
If $\lambda^{(k)} > \lambda^{(k + 1)}$, then the error exponent of the sample complexity is
  \begin{align} 
   -\lim_{\eps \to 0} \frac{1}{\eps^2} \lim_{n \to \infty} \frac{1}{n} \log \mathbb{P}_n \left\{ \bbfrob{ \dtmt \Psib }^2 - \bbfrob{ \dtmt \Psibh }^2 > \eps^2 \right\} = \frac{1}{2\alpha_k},
    \label{eq:exponent}
  \end{align} 
  where $\alpha_k$ is the largest singular value of the matrix $\Gb^{\frac{1}{2}}_k \Jb \Gb_k^{\frac{1}{2}}$, 
where $\Gb_k$ is defined as
\begin{align}  \label{eq:gb:k}
  \Gb_k \defeq \sum_{i = 1}^k\sum_{j = k + 1}^{m} \frac{
\left(\psib^{(j)} \circ \psib^{(i)}\right)\left(\psib^{(j)} \circ \psib^{(i)}\right)^{\T}
  }{\lambda^{(i)} - \lambda^{(j)}}
\end{align}
with
\begin{align}
  \psib^{(j)} \circ \psib^{(i)} = 
  \begin{bmatrix}
    \psib_1^{(j)}\otimes\psib_1^{(i)}\\
    \psib_1^{(j)}\otimes\psib_2^{(i)}\\
    \vdots\\
    \psib_d^{(j)}\otimes\psib_d^{(i)}
  \end{bmatrix}
  \label{eq:tracy}
\end{align}
denoting the Tracy-Singh product \cite{tracy1972new} of $\psib^{(j)}$ and $\psib^{(i)}$ as partitioned in \eqref{eq:Bmatrixa}, 
and where $\Gb_k^{\frac{1}{2}}$ is the positive semidefinite matrix $\Mb$ such that
$\Mb^2 = \Gb_k$.%

\end{theorem}
\begin{proof}
  See Appendix \ref{sec:a:thm:exponent}.
\end{proof}

From Theorem \ref{thm:exponent}, the sample complexity depends only on joint distribution of all quadruples among $X^d$.

\fi

\section{Connections to Existing Techniques} \label{sec:5}

In this section, we demonstrate the relationship between our functional representations and the Hirschfeld-Gebelein-R\'{e}nyi (HGR) maximal correlation~\cite{renyi}, linear PCA~\cite{Jolliffe}, and the consistent functional map~\cite{Huang:2014}. This  demonstrates the deep connections between our approach and existing techniques, while offering novel information theoretic interpretations to machine learning algorithms.

\subsection{The HGR maximal correlation} 

The HGR maximal correlation is a variational generalization of the well-known Pearson correlation coefficient, and was originally introduced as a normalized measure of the dependence between two random variables~\cite{renyi}.
\begin{definition}[Maximal Correlation] 
\label{Def:Maximal Correlation}
For jointly distributed random variables $X$ and $Y$, with discete ranges $\cX$ and $\cY$ respectively, the maximal correlation between $X$ and $Y$ is defined as:
\begin{align*}
\rho(X;Y) \triangleq \max_{\substack{{f\colon \cX \mapsto \mathbb{R}, \enspace g\colon \cY \mapsto \mathbb{R} %
  } }} \E\left[f(X)g(Y)\right]
\end{align*}
where the maximum is taken over zero-mean and unit-variance functions $f(X)$ and $g(Y)$.
\end{definition}
The HGR maximal correlation has been shown useful not only as a statistical measurement, but also in designing machine learning algorithms for regression problems~\cite{SL15}\cite{SL17}\cite{FT15}. To draw the connection, note that in the bivariate case $d = 2$, the functions derived in Section~\ref{sec:2} are precisely the maximal correlation functions for two random variables. In addition, our functional representation for general cases essentially defines a generalized version of the maximal correlation [cf.~\eqref{eq:MHGRk}].

\begin{definition} \label{def:gmc}
The generalized maximal correlation for jointly distributed random variables $X_1 , \ldots , X_d$ with discrete ranges $\cX_i$, for $i = 1, \ldots, d$, is defined as
\begin{align} \label{eq:MMC}
\rho^* (X_1 , \cdots , X_d) \triangleq \max \frac{1}{d-1}  \E \left[ \sum_{i \neq j} f_i(X_i) f_j(X_j) \right] 
\end{align}
for the functions $f_i : \cX_i \mapsto  \mathbb{R}$, with the constraints $\E \left[ f_i(X_i) \right] =  0, \ \E \left[ \sum_{i = 1}^d f_i^2(X_i) \right] = 1$, for all $i$.
\end{definition}
It is easy to verify that $0 \leq \rho^* (X_1 , \cdots , X_d) \leq 1$, and $\rho^* (X_1 , \cdots , X_d) = 0$ if and only if $X_1 , \ldots , X_d$ are pairwise independent.

Note that there are some other generalizations to maximal correlations to multiple random variables. For example, the network maximal correlation (NMC) proposed in~\cite{Ali} defined a correlation measurement in the same way as~\eqref{eq:MMC} but with a slightly different constraint: $\E \left[ f_i(X_i) \right] =  0, \ \E \left[ f_i^2(X_i) \right] = 1, \ \text{for all $i$.}$
In addition, \cite{SD} proposed a maximally correlated principal component analysis, which considered the SVD of a matrix similar to $\dtm$. However, our approach and results essentially offer the information theoretic justification of generalizing the maximal correlation as extracting common structures shared among random variables, and also provide the guidance to algorithm designs.

\subsection{Linear PCA}

It turns out that the functional representation derived in Section~\ref{sec:2} is a nonlinear generalization to the linear PCA~\cite{Jolliffe}. To see that, consider a sequence of data vectors $\ux^{(\ell)} = ( x^{(\ell)}_1 , \ldots , x^{(\ell)}_d ) \in \mathbb{R}^d$, for $\ell = 1, \ldots, n$, where the sample mean and variance for each dimension are zero and one, respectively, i.e., $\sum_{\ell = 1}^n x_i^{(\ell)}= 0$, and $\frac{1}{n} \sum_{\ell = 1}^n  (x_i^{(\ell)})^2 =1$, for all $i$. Then, the linear PCA aims to find the principle vector $\uw = ( w_1 , \ldots , w_d )$ with unit norm such that $\sum_{\ell = 1}^n \langle \uw , \ux^{(\ell)} \rangle^2$ is maximized; or equivalently, to maximize
\begin{align} \label{eq:PCA}
\frac{1}{n}\sum_{\ell = 1}^n \sum_{i \neq j} \left( w_i x^{(\ell)}_i  \right) \left( w_j x^{(\ell)}_j \right) = \E \left[ \sum_{i \neq j} \left( w_i X_i  \right) \cdot \left( w_j X_j \right) \right]
\end{align}
subject to the constraint 
\begin{align} \label{eq:PCA2}
1 = \sum_{i = 1}^d w_i^2 = \sum_{i = 1}^d \E \left[ \left( w_i X_i \right)^2 \right],
\end{align}
where the expectations in~\eqref{eq:PCA} and~\eqref{eq:PCA2} are taking over the empirical distributions $P_{X_i X_j}$ and $P_{X_i}$ from the data vectors. Comparing to the Definition~\ref{def:gmc}, we can see that our functional representation generalizes the linear PCA to nonlinear functional spaces of data. We would like to emphasize that~\cite{SL17} also provides a nonlinear generalization to PCA for the Gaussian distributed data vectors by the local geometric approach. Our approach presented in this paper essentially offers another generalization for general discrete data vectors.

\subsection{Consistent Functional Map}

In computer vision, a typical question is to find the shared components among a collection of shapes, for example, the legs of chairs, when some noisy maps between these shapes are given. An effective approach to extract such shared structure between shape collections, called consistent functional map network, is recently proposed in~\cite{Huang:2014, Wang:2013, Wang:2014}. The main idea of the consistent functional map network is to formulate the shared components as low-dimensional subspaces of the functional spaces of these shapes, and the given noisy maps between these shapes are formulated as the transition maps between these functional spaces. Then, the goal of the consistent functional map network is to find a low-dimensional subspace of the functional space of each shape, such that under a cycle of transition maps between shapes, this low-dimensional subspace remains the same. 

Note that this idea is similar to the functional representation we derived in this paper, except that the consistent functional map network considers the transition maps between shapes that are deterministic maps, while we consider stochastic maps between random variables. In fact, it is shown in~\cite{Huang:2014} that if we write the noisy maps between shapes $i$ and $j$ as $\Mb_{ij}$, then such subspaces can be solved by the eigen-decomposition of a matrix by replacing the stochastic transition map $\dtm_{ij}$ of $\dtm$ in~\eqref{eq:Bmatrix} into the noisy (deterministic) maps $\Mb_{ij}$ between shapes. Therefore, the functional representation presented in this paper can be viewed as an extension of the consistent functional map network to general stochastic object, which can essentially be applied to a wider range of problems.

\section{The Numerical Simulations}

The functional representations of the data can be viewed as low-dimensional feature functions selected from the hidden common structure. In this section, we show that such selected low-dimensional feature functions can practically be useful by verifying the performance in the MNIST Handwritten Digit Database~\cite{MNIST} for digits recognition. In the MNIST database, there are $n = 60\,000$ images contained in the training sets, and each image has a label that represents the digits ``0'' to ``9''. The images in this database are consisted of $28 \times 28$ pixels, where each image pixel takes the value ranging from 0 to 255. While this is a supervised learning problem, we will show that both Algorithm~\ref{alg:2} and the low-rank approximation method described in Section \ref{sec:low-rank} can be applied to select features from images directly without the knowledge of labels, and these features, although selected in an unsupervised way, have good performance in handwritten digit recognition. 

To begin, we need to identify the random variables $X_i$ in the MNIST problem. For this purpose, we divide each image into $8 \times 8 = 64$ overlapping subareas, where each sub-image has $6 \times 6$ pixels, and two nearby subareas are overlapped with $3$ pixels. Fig.~\ref{fig:MNIST} illustrates this division of images.

The purpose of dividing entire image into subareas is to reduce the complexity of training joint feature functions among image pixels, while capturing the correlations between nearby pixels. Then, each sub-image $i$ of the $64$ subareas can be viewed as a random variable $X_{i}$, for $i = 1, \ldots , 64$. Therefore, if we denote $x^{(\ell)}_{i}$ as the value of the sub-image $i$ of the $\ell$-th image of the MNIST database, then each random variable $X_i$ has $n$ training samples $x^{(1)}_i , \ldots, x^{(n)}_{i}$.

\subsection{Apply the MACE Algorithm to MNIST}
To apply the MACE Algorithm~\ref{alg:2}, we further quantize each image pixel into binary signals ``0'' and ``1'' with the quantization threshold 40.
Note that each $x^{(\ell)}_i$ is essentially a $36$-dimensional binary vector, thus the cardinality of the alphabet $| \X_i | = 2^{36}$. To reduce the cardinality, for each subarea $i$, we go through $n$ training images to find all possible binary vectors in $\{0,1\}^{36}$, and then map these binary vectors into a smaller alphabet set, such that two binary vectors with Hamming distance no greater than three are mapped into the same alphabet. This quantization procedure is illustrated in Algorithm~\ref{alg4}.

\begin{algorithm} [t]
\caption{Quantizing Alphabets to Reduce the Cardinality}
\label{alg4}
\begin{algorithmic}[]
\Require{training samples $\left\{x^{(\ell)}_i\colon \ell=1, \dots, n\right\}$}
\\ Initialize: set the alphabet $\X_{i} \gets \varnothing$.\\
\textbf{For} $\ell = 1:n$ \\
$\quad$ \textbf{If} $\exists\, x \in \X_{i}$, such that $d_H\left(x, x^{(\ell)}_i\right) \leq 3$. \\
$\quad$ $\quad$ \textbf{Then} set $x^{(\ell)}_i \gets x$. \\
$\quad$ \textbf{Else} $\X_i \gets \X_{i} \cup \left\{ x^{(\ell)}_i \right\}$ \\
\textbf{End}
\end{algorithmic}
\end{algorithm}

After this pre-processing, 64 random variables $X_{i}$ are specified, and each image $\ell$ can be viewed as a 64-dimensional data vector $(x^{(\ell)}_1 , \ldots , x^{(\ell)}_{64})$, for $\ell = 1, \ldots , n$. Then, we apply Algorithm~\ref{alg:2} to compute $k$ feature functions $\vec{f}_{i} = (f_{i}^{(1)}, \ldots , f_{i}^{(k)})$  for each random variable $X_{i}$. These feature functions map the pre-processed training image $\ell$ into a $(64k)$-dimensional score vector 
\begin{align*}
\vec{s}_{\ell} = \left( \vec{f}_1(x^{(\ell)}_1), \ldots , \vec{f}_{64}(x^{(\ell)}_{64}) \right),
\end{align*}
which extracts non-linear features of the image. Note that in this step, we select the feature functions only from the image pixels but without the knowledge of the labels.  

With the score vectors computed, at the second step we apply the linear support vector machine (SVM)~\cite{SVM} to classify the vectors $\vec{s}_\ell$, for $\ell = 1, \ldots, n$ into ten groups with respect to the labels $z_\ell$. This results in a linear classifier that associates a label $\hat{z}_\ell \in \{ 0, \ldots , 9 \}$ to each score vector $\vec{s}_\ell$, and the label represents the recognized digit of the image corresponding to the score vector.

To test the performance of this linear classifier in the set of test images, we first conduct the same pre-processing to the test images, and map the pre-processed test images into $(64k)$-dimensional score vectors by the feature functions $\vec{f}_{i}$. Then, the linear classifier is applied to recognize the digits in the test images, and the error probabilities of recognizing the digits via the score vectors with different values of $k$ are demonstrated in the following table.

\begin{center} \label{t1}
  \begin{tabular}{| c | c | c | c | c | c | c |}
    \hline
    $k$ & 4  & 8  & 12 & 16 & 20 & 24 \\ \hline
    Error rate ($\%$)  & 4.74 & 2.44 & 2.36 & 2.21 & 2.15 & 2.08   \\ \hline
  \end{tabular}
\end{center}

Note that our approach can be viewed as mapping the image pixels to the feature space by \emph{one layer} of informative score functions and then apply the linear classification. It turns out that the error rate of our approach is comparable to the neural networks with \emph{two layers} of feature mapping by the sigmoid functions (the error rate is 2.95\% for a 3-layer fully-connected neural network with 500 and 150 units in two hidden layers~\cite{MNIST, NN}). Moreover, the neural networks select the features with the aid of labels, while the feature functions in our approach are selected without the knowledge of label but from the shared structure between subareas. This essentially shows how the information from shared structures can be applied to practical problems by our algorithms.

\subsection{Finding Functional Representations by Neural Networks} \label{sec:simu_NN}

As illustrated in Section \ref{sec:low-rank}, we first use 64 neural networks $\mathtt{NN}_1, \dots, \mathtt{NN}_{64}$ to generate representations $\uf_1(X_1), \dots, \uf_{64}(X_{64})$ from images, where each neural network $\mathtt{NN}_i$ consists of two convolutional layers as shown in Fig. \ref{fig:cnn}. Using the negative MH-score $-H\left(\uf_1(X_1), \dots, \uf_{64}(X_{64})\right)$ as the loss function, we then train these 64 neural networks to obtain the optimal functional representations.

\begin{figure}[!t]
    \centering
    \resizebox{.37\textwidth} {!}{\tikzset{%
  block/.style    = {draw = black, thick, rectangle, minimum height = 1.5em, minimum width = 5em, text width = 8.5em, align = center},
  blockc/.style    = {block, fill = gray!50!white},
  input/.style    = {coordinate}, 
  output/.style   = {coordinate} 
}
\newcommand{\suma}{\Large$+$}

\def\blkdist{.3cm}
\begin{tikzpicture}[auto, thick, node distance=1.8cm, ->,draw=white!60!black!20!blue]
  
  \draw
  node [block, text width = 12em, draw = none] (blk1) {Input Image $X_i$: 6 $\times$ 6}
  node [blockc, above = \blkdist of blk1] (blk2) {Conv: 2 $\times$ 2 $\times$ 16}
  node [blockc, above = \blkdist of blk2] (blk3) {Conv: 2 $\times$ 2 $\times$ 32}
  node [block, above = \blkdist of blk3] (blk4) {Max Pooling: 2 $\times$ 2}
  node [block, above = \blkdist of blk4] (blk5) {Dropout (0.25)}  
  node [block, above = \blkdist of blk5] (blk6) {Flatten}  
  node [blockc, above = \blkdist of blk6] (blk7) {Fully Connected 
  }
  node [block, draw = none, above = \blkdist of blk7] (out) {$k$-dimensional Output: $\uf_i(X_i)$};
  \draw[->](blk1) -- (blk2);
  \draw[->](blk2) -- (blk3);
  \draw[->](blk3) -- (blk4);
  \draw[->](blk4) -- (blk5);
  \draw[->](blk5) -- (blk6);
  \draw[->](blk6) -- (blk7);  
  \draw[->](blk7) -- (out);

  \begin{pgfonlayer}{background}
    \node[draw, dashed,fit=(blk2) (blk3) (blk4) (blk5) (blk6) (blk7), label={[xshift=6em, yshift=-15em]$\mathtt{NN}_i$}] (A) {};
  \end{pgfonlayer}
\end{tikzpicture}

}
    \caption{The architecture of the $i$-th neural network $\mathtt{NN}_i$ %
      that extracts feature $\uf_i(X_i)$ from the input $X_i$.}
  \label{fig:cnn}
\end{figure}
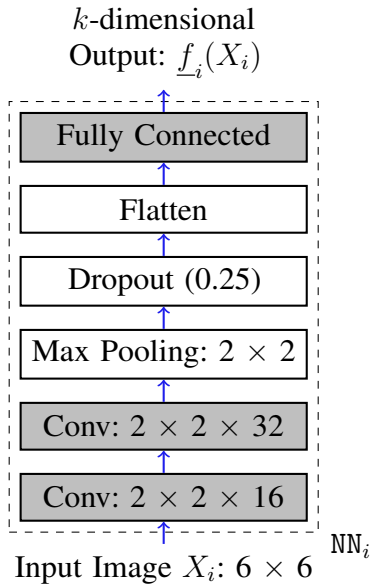

With the functional representations trained from the training set, we again adopt the linear SVM for the classification task and use this linear classifier to recognize the test images. The following table shows the classification error with different values of $k$.
\begin{center} \label{t2}
  \begin{tabular}{| c | c | c | c | c | c | c |}
    \hline
    $k$ & 4  & 8  & 12 & 16 & 20 & 24 \\ \hline
    Error rate ($\%$)  & 3.46 & 1.73 & 1.43 & 1.17 & 1.15 & 1.11   \\ \hline
  \end{tabular}
\end{center}

Compared with the results from the MACE Algorithm \ref{alg:2}, the neural network approximation method has better performance. This performance gain is mainly from directly processing the subareas of images by the CNNs without the information loss in the quantization step. In addition, the features extracted in this unsupervised approach can achieve the performance comparable to the CNN-based supervised learning algorithms, such as LeNet-4, which has an error rate of 1.1\% \cite{NN}.

\appendices
\section{Proof of Lemma \ref{lem:1}}
\label{sec:a:lem:1}
To show the first property, we define the $|\cX_i | \times (|\cX_1| \cdot |\cX_2| \cdots |\cX_d|)$ matrix $\dtm_i$, for $i = 1, \ldots , d$, as
\begin{align*}
B_i (x'_i; x^d) = 
\begin{cases}
\frac{\sqrt{P_{X^d}(x^d)} }{\sqrt{P_{X_i}(x'_i)}}, & \text{if} \ x'_i = x_i, \\
0, & \text{otherwise.}
\end{cases}
\end{align*}
Then, one can verify that
\begin{align*}
  \dtm =
  \begin{bmatrix}
    \dtm_1\\
    \vdots\\
    \dtm_d
  \end{bmatrix}
  \begin{bmatrix}
    \dtm_1^{\T}&
    \cdots&
    \dtm_d^{\T}
  \end{bmatrix},
\end{align*}
which implies that $\dtm$ is positive semidefinite.

To establish the second property, note that $\psib^{(0)}$ is an eigenvector of $\dtm$ with eigenvalue $d$, and thus
  $\left(\psib^{(0)}\right)^\T \dtm \psib^{(0)} = d$.
Moreover, it is shown in~\cite{SL12} that the largest singular value of $\dtm_{ij}$ is 1, i.e., $\bbspectral{\dtm_{ij}} = 1$, where $\bbspectral{\cdot}$ denotes the spectral norm of its matrix argument. 

Therefore, for $\psib = \left[\psib_1^\T, \dots, \psib_d^\T\right]^\T$
with each $\psib_i$ being an $|\cX_i|$-dimensional vector, we have
\begin{align*}
  \psib^\T\dtm\psib
    = \sum_{i = 1}^d\sum_{j = 1}^d\psib_i^\T\dtm_{ij}\psib_j%
  &\leq \sum_{i = 1}^d\sum_{j = 1}^d \|\psib_i\|\cdot \bbspectral{\dtm_{ij}}\cdot \|\psib_j\|%
    = \left(\sum_{i = 1}^d \|\psib_i\|\right)^2\\  
  &\leq d \sum_{i = 1}^d \|\psib_i\|^2%
    = d \|\psib\|^2,
\end{align*}
where the second inequality follows from the fact that the arithmetic mean is no greater than the quadratic mean.
Hence, we have
\begin{align*}
  \max_{\psib\colon \|\psib\| = 1} \psib^\T\dtm\psib = d,
\end{align*}
i.e., the largest eigenvalue of $\dtm$ is $d$.

To verify the third property, we construct $\psib'$ as
\begin{align*}
  \psib' =
  \begin{bmatrix}
    \psib'_1\\
    \zerob
  \end{bmatrix},
\end{align*}
where $\psib'_1 \in \mathbb{R}^{|\cX_1|}$ is chosen such that $\langle\psib'_1, \vb_1\rangle = 0$ and $\|\psib'_1\|^2 = 1$, and where $\zerob$ is the $(m-|\cX_1|)$-dimensional zero vector. Therefore, we have
  $\left\langle\psib', \psib^{(0)}\right\rangle = 0$ and
  $\left\|\psib'\right\|^2 = 1$.
Note that the second eigenvalue $\lambda^{(1)}$ of $\dtm$ can be written as
\begin{align*}
  \lambda^{(1)} =  \max_{\psib\colon \|\psib\| = 1, \left\langle\psib, \psib^{(0)}\right\rangle = 0} \psib^\T\dtm\psib,
\end{align*}
which implies that
  $\lambda^{(1)} \geq \left(\psib'\right)^\T\dtm\psib' = \|\psib_1'\|^2 = 1$.

To verify the fourth property, we define the $(d-1)$-dimensional subspace
$\cS_{\mathrm{eig}}$ as
\begin{align}
  \cS_{\mathrm{eig}} \defeq
  \left\{
  \psib = \left[\alpha_1 \vb_1^{\T} , \ldots , \alpha_d \vb_d^{\T}\right]^{\T}\colon \sum_{i = 1}^d \alpha_i = 0
  \right\}.
\end{align}
Then, for all $\psib \in \cS_{\mathrm{eig}}$, from $\dtm_{ij}\vb_j = \vb_i$, it is straightforward to verify that $\dtm\psib = \zerob_m$, where $\zerob_m$ is the zero vector in $\mathbb{R}^m$. Therefore, $\cS_{\mathrm{eig}}$ is an eigenspace of $\dtm$ associated with $d - 1$ zero eigenvalues. Since $\dtm$ is positive semidefinite, without loss of generality we can assume that $\cS_{\mathrm{eig}}$ is spanned by $\psib^{(m - d + 1)}, \dots, \psib^{(m - 1)}$, which correspond to eigenvalues $\lambda^{(m - d + 1)} = \cdots = \lambda^{(m - 1)} = 0$.

Finally, to establish the last property, for each $\ell = 1, \dots, m - d$, from $\left\langle\psib^{(\ell)}, \psib^{(0)}\right\rangle = 0$ we have
\begin{align*}
  \sum_{i = 1}^d \left\langle\psib_i^{(\ell)}, \vb_i\right\rangle = 0.
\end{align*}
Therefore, from the third property, we have
\begin{align*}
  \psib' =
  \begin{bmatrix}
    \left\langle\psib_1^{(\ell)}, \vb_1\right\rangle \vb_1\\
    \vdots\\
    \left\langle\psib_d^{(\ell)}, \vb_d\right\rangle \vb_d\\    
  \end{bmatrix}
  \in \cS_{\mathrm{eig}}.
\end{align*}
Hence, we obtain $\langle\psib', \psib^{(\ell)}\rangle = 0$, i.e.,
\begin{align*}
  \sum_{i = 1}^d \left\langle\psib_i^{(\ell)}, \vb_i\right\rangle^2 = 0,
\end{align*}
which implies that
$ \left\langle\psib_i^{(\ell)}, \vb_i\right\rangle = 0$ for %
$i = 1, \dots, d $.
%

\section{Proof of Theorem \ref{thm:1}}
\label{sec:a:thm:1}
First, we replace $\delta$ by $\frac{1}{2}\epsilon^2$ for the convenience of presentation when applying the local geometric approach. With this notation, the constraint~\eqref{eq:total_correlation2} becomes 
\begin{align} \label{eq:eps^2/2}
I(U; X^d) \leq \frac{1}{2}\epsilon^2, 
\end{align}
with $\epsilon$ assumed to be small. Then, it follows from~\eqref{eq:no_tiltied} and~\eqref{eq:eps^2/2} that for all $u$, the conditional distribution $P_{X^d | U = u}$ can be written as a perturbation to the marginal distribution:
\begin{align} \label{eq:pert_vec}
P_{X^d | U} (x^d | u) =  %
P_{X^d} (x^d)  + \epsilon \sqrt{P_{X^d} (x^d)} \phi_u (x^d)
\end{align}
where $\phi_u$ can be viewed as an $(|\cX_1| \cdot |\cX_2| \cdots |\cX_d|)$-dimensional vector. Moreover, it follows from the second order Taylor's expansion for the K-L divergence that
\begin{align*} %
I(U; X^d) 
&= \E_U \left[ D(P_{X^d | U} \| P_{X^d} ) \right] %
  =  \frac{1}{2}\epsilon^2 \E_U \left[ \left\| \phib_U \right\|^2 \right] + o(\epsilon^2),
\end{align*}
where $\| \cdot \|$ denotes the $l_2$-norm. Thus, by ignoring the higher order term of $\epsilon$ as we assume $\epsilon$ to be small, the constraint $I(U; X^d) \leq \frac{1}{2}\epsilon^2$ can be reduced to
\begin{align} \label{eq:aaa}
 \E_U [\| \phib_U \|^2] \leq 1.
\end{align}

In addition, the objective function $\ell( X^d | U)$ can also be expressed in terms of mutual informations:
\begin{align} %
&D(P_{X^d} \| P_{X_1} \cdots P_{X_d}) - D(P_{X^d} \| P_{X_1} \cdots P_{X_d} | U) %
  = \sum_{i=1}^d I(U;X_i) - I(U; X^d) \label{eq:CI}
\end{align}
and for each $i$, the mutual information $I(U;X_i)$ can be again approximated as the $l_2$-norm square
\begin{align*} %
I(U;X_i) = \frac{1}{2}\epsilon^2  \E_U [\| \psib_{i,U} \|^2] + o(\epsilon^2),
\end{align*}
where for $U=u$, the vector $ \psib_{i,u}$ is the $|\cX_i |$-dimensional perturbation vector defined as
\begin{align} \label{eq:xp}
\psi_{i,u} (x_i) = \frac{P_{X_i|U}(x_i | u) - P_{X_i} (x_i)}{\epsilon \sqrt{P_{X_i}(x_i)}}
\end{align}

Then, by ignoring the higher order terms of $\epsilon$, the optimization problem we want to solve can be transferred to a linear algebraic problem
\begin{align} \label{eq:problem2}
\max_{\E_U [\| \phib_U \|^2] \leq 1} \  \sum_{i=1}^d \E_U [\| \psib_{i,U} \|^2] - \E_U [\| \phib_U \|^2].
\end{align}

To solve~\eqref{eq:problem2}, observe that $P_{X_i}$ and $P_{X_i | U}$ are marginal distributions of $P_{X^d}$ and $P_{X^d | U}$, thus there is a correlation between $\phib_U$ and $\psib_{i,U}$:
\begin{align*} %
\psi_{i,u} (x_i) = %
\sum_{x_1, \ldots , x_{i-1}, x_{i+1}, \ldots , x_d} \frac{\sqrt{P_{X^d} (x^d)}}{\sqrt{P_{X_i}(x_i)}} \phi_u (x^d)
\end{align*}
which can be represented in matrix form as $\psib_{i,u} = \dtm_i \cdot \phib_u$, where $\dtm_i$ is an $|\cX_i | \times (|\cX_1| \cdot |\cX_2| \cdots |\cX_d|)$ matrix with entries
\begin{align}
  B_i (x'_i; (x_1 , \ldots , x_d)) = 
  \left\{
  \begin{array}{cc}
    \frac{\sqrt{P_{X^d}(x^d)} }{\sqrt{P_{X_i}(x'_i)}} & \text{if} \ x_i' = x_i, \\
    0  & \text{otherwise.}
  \end{array}
  \right.
  \label{eq:def:Bi}
\end{align}
Therefore, if we define an $(|\cX_1| + \cdots + |\cX_m|)\times (|\cX_1|  \cdots  |\cX_m|)$-dimensional matrix
\begin{align}
  \dtm_0 \triangleq
  \begin{bmatrix}
    \dtm_1\\
    \vdots\\
    \dtm_d
  \end{bmatrix},
  \label{eq:B0}
\end{align}
then since
\begin{align*}
    \sum_{i=1}^d \E_U [\| \psib_{i,U} \|^2] = \sum_{i=1}^d \E_U [\| \dtm_i \cdot \phib_U \|^2] = \E_U [\| \dtm_0 \cdot \phib_U \|^2],
\end{align*}
we can rewrite \eqref{eq:problem2} as
\begin{align} \label{eq:problem_linear}
  \max_{\E_U [\| \phib_U \|^2] \leq 1}  \E_U [\| \dtm_0 \cdot \phib_U \|^2] - \E_U [\| \phib_U \|^2].
\end{align}
Moreover, since $\phib_U$ is a perturbation vector of probability distributions, by summing over all $x^d$ %
for both sides of~\eqref{eq:pert_vec}, it has to satisfy an extra constraint
\begin{align} %
  \sum_{x^d} \sqrt{P_{X^d} (x^d)} \phi_u (x^d) = 0,
  \label{eq:ortho:phi0}
\end{align}
which implies that $\phib_U$ is orthogonal to an $\left(|\cX_1| \cdot |\cX_2| \cdots |\cX_d|\right)$-dimensional vector $\phib^{(0)}$, whose entries are $\sqrt{P_{X^d}(x^d)}$. In particular, it is shown in~\cite{SL12} that $\phib^{(0)}$ is the right singular vector of $\dtm_0$ with the largest singular value $\sigma_0 = \sqrt{d}$, and the corresponding left singular vector is $\psib^{(0)}$. %
In addition, it can be verified that $\dtm_0$ satisfies $\dtm_0 \dtm_0^\T = \dtm$ with $\dtm$ as defined in \eqref{eq:Bmatrix}. Therefore, the second largest singular value of $\dtm_0$ is $\sigma_1 = \sqrt{\lambda^{(1)}} \geq 1$, and the optimal solution of~\eqref{eq:problem_linear} is to align the vectors $\phib_{U=u}$, for all $u$, along the second largest right singular vector of $\dtm_0$. 

It turns out that it is easier to compute the second largest left singular vector of $\dtm_0$ instead of the right one. This is equivalent to computing the second largest eigenvector of the matrix $\dtm_0 \dtm_0^\T = \dtm$.
%

%
%

Now, %
the second largest right singular vector $\phib^{(1)}$ of $\dtm_0$ can be computed as
\begin{align} %
\phi^{(1)} (x^d) 
  = \frac{1}{\sqrt{\lambda^{(1)}}} (\dtm_0^\T \psib^{(1)}) (x^d) %
&=  %
  \frac{1}{\sqrt{\lambda^{(1)}}} \cdot
  \left( \sqrt{P_{X^d} (x^d)}\sum_{i = 1}^d \frac{\psi_i^{(1)}(x_i)}{\sqrt{P_{X_i}(x_i)}} \right) \notag\\
&= \sqrt{P_{X^d} (x^d)} \cdot \left( \frac{1}{\sqrt{\lambda^{(1)}}} \sum_{i = 1}^d f^{(1)}_i(x_i) \right),\label{eq:phi:1}
\end{align}
where $\dtm_0^\T \psib^{(1)}$ is a vector and $(\dtm_0^\T \psib^{(1)}) (x^d)$ is the $x^d$%
-th entry of this vector. Since all the $\phib_{U=u}$ should be aligned to $\phib^{(1)}$, there exists a function $h : {\cal U} \mapsto \mathbb{R}$, such that
\begin{align*} %
P_{X^d |U} (x^d |u) 
= P_{X^d} (x^d)  \left( 1+  \frac{\epsilon h(u)}{\sqrt{\lambda^{(1)}}} \sum_{i = 1}^d f_i^{(1)} (x_i) \right) + o(\epsilon),
\end{align*}
where the term $o(\epsilon)$ comes from the local approximation we made for~\eqref{eq:aaa}. Therefore, the optimal joint distributions for our optimization problem can be written as
\begin{align} %
&P_{U X^d} (u, x^d) %
  = P_U(u) P_{X^d} (x^d)  \left( 1+ \frac{\epsilon h(u)}{\sqrt{\lambda^{(1)}}} \sum_{i = 1}^d f_i^{(1)} (x_i) \right) + o(\epsilon).\label{eq:bbb}
\end{align}
Note that if we sum both sides of~\eqref{eq:bbb} over all $u \in {\cal U}$, then we have $\sum_{u \in {\cal U}}P_U(u)h(u) = 0$, which implies that $h(U)$ is a zero-mean function. Moreover, it is easy to compute from~\eqref{eq:aaa} that the variance $\E [h^2(U)] = 1$. Finally, note that the exponential family ${\cal P}^{(\delta)}_{\exp}$, when $\delta$ is small, can be written as
\begin{align} %
{\cal P}^{(\delta)}_{\exp} &= \Bigg\{P_U(u) P_{X^d} (x^d)  %
  \cdot \left( 1+ \frac{\sqrt{2\delta} h(u)}{\sqrt{\lambda^{(1)}}} \sum_{i = 1}^d f_i^{(1)} (x_i) \right) + o\left(\sqrt{\delta}\right) : h \in {\cal H}_{\delta} \Bigg\}.\label{eq:ccc}
\end{align}
Since $\delta = \frac{1}{2} \epsilon^2$ the proof is completed by comparing~\eqref{eq:ccc} and~\eqref{eq:bbb}. 
\section{Proof of Theorem \ref{thm:2}}
\label{sec:a:thm:2}
First, we introduce a useful lemma (see, e.g., \cite[Corollary 4.3.39, p. 248]{horn2012matrix}).
\begin{lemma}
  \label{lem:frob}
  Given an arbitrary $k_1 \times k_2$ matrix $\Ab$ and any $k \in \bigl\{1, \dots, \min\{k_1, k_2\}\bigr\}$, we have
  \begin{align}
    \max_{\Mb \in \mathbb{R}^{k_2 \times k}} \bfrob{\Ab\Mb}^2 = \sum_{i = 1}^k \sigma_i^2,
    \label{eq:lem:frob}
  \end{align}
  where $\frob{\cdot}$ denotes the Frobenius norm, and where $\sigma_1 \geq \dots \geq \sigma_{\min\{m, n\}}$ denotes the singular values of $\Ab$. Moreover, the maximum in \eqref{eq:lem:frob} can be achieved by
    $\Mb =
    \begin{bmatrix}
      \vb_1 & \cdots & \vb_k
    \end{bmatrix}
    \Qb$,
  where $\vb_i$ denotes the right singular vector of $\Ab$ corresponding to $\sigma_i$, for $i = 1, \dots, \min\{m, n\}$, and $\Qb \in \mathbb{R}^{k\times k}$ is an orthogonal matrix.
\end{lemma}

To begin the proof, similar to Theorem \ref{thm:1}, we replace $\delta$ by $\frac{1}{2}\epsilon^2$ and write the conditional distribution $P_{X^d|U^k = u^k}$ as a perturbation to the joint distribution $P_{X^d}$:
  \begin{align}
    P_{X^d|U^k}(x^d|u^k) = P_{X^d}(x^d) + \eps\sqrt{P_{X^d}(x^d)}\phib_{u^k}(x^d),
    \label{eq:phi:uk}
  \end{align}
  where $\phib_{u^k}$ is an $(|\cX_1|\cdot|\cX_2|\cdots|\cX_d|)$-dimensional vector. Again, it follows from the second-order Taylor series expansion of the K-L divergence that
  \begin{align*}
    I(U^k;X^d)
    &= \E_{U^k}[D(P_{X^d|U^k}\|P_{X^d})]%
      = \frac{1}{2}\eps^2 \E_{U^k}\left[\left\|\phib_{U^k}\right\|^2\right] + o(\eps^2).
  \end{align*}

  Similarly, for all $i = 1, \dots, k$, the conditional distribution $P_{X^d|U_i = u_i}$ can be written as
  \begin{align}
    P_{X^d|U_i}(x^d|u_i) = P_{X^d}(x^d) + \eps\sqrt{P_{X^d}(x^d)}\phib_{u_i}(x^d),
    \label{eq:phi:ui}
  \end{align}
  and we have
  \begin{align*}
    I(U_i;X^d)
    = \frac{1}{2}\eps^2 \E_{U_i}\left[\left\|\phib_{U_i}\right\|^2\right] + o(\eps^2).
  \end{align*}
  Therefore, by ignoring the higher order terms of $\eps$, the first constraint %
  can be reduced to
  \begin{align*}
    1 \geq \E_{U_1}\left[\left\|\phib_{U_i}\right\|^2\right] \geq \dots \geq \E_{U_k}\left[\left\|\phib_{U_i}\right\|^2\right].
  \end{align*}

  Moreover, due to the independence and conditional independence among the $U^k$, $\phib_{u^k}$ and $\phib_{u_i}$ satisfy
  \begin{align}
    \phib_{u^k} = \sum_{i = 1}^k \phib_{u_i} + o(1)
    \label{eq:phib:uk}
  \end{align}
  and
  \begin{align}
    \langle\phib_{u_i}, \phib_{u_j}\rangle = 0, \quad \text{for all~}i \neq j, u_i \in \cU_i, u_j \in \cU_j.%
    \label{eq:phib:uij}
  \end{align}
  Indeed, we have
  \begin{align*}
    P_{X_d|U^k}(x^d|u^k)
    &=  \frac{P_{X^d}(x^d)P_{U^k|X^d}(u^k|x^d)}{P_{U^k}(u^k)}%
      = P_{X^d}(x^d) \prod_{i = 1}^k \frac{P_{U_i|X^d}(u_i|x^d)}{P_{U_i}(u_i)},
  \end{align*}
  which implies
  \begin{align}
    \frac{P_{X_d|U^k}(x^d|u^k)}{P_{X^d}(x^d)}
    = \prod_{i = 1}^k \frac{P_{U_i|X^d}(u_i|x^d)}{P_{U_i}(u_i)}
    = \prod_{i = 1}^k \frac{P_{X^d|U_i}(x^d|u_i)}{P_{X^d}(x^d)}.
    \label{eq:p:xd:uk}
  \end{align}
  Substituting \eqref{eq:phi:uk} and \eqref{eq:phi:ui} into \eqref{eq:p:xd:uk} then yields
  \begin{align*}
    1 + \eps \frac{\phib_{u^k}(x_1^d)}{\sqrt{P_{X^d}(x^d)}}
    = \prod_{i = 1}^k \left[1 + \eps \frac{\phib_{u_i}(x_1^d)}{\sqrt{P_{X^d}(x^d)}}\right],
  \end{align*}
  and via comparing the $\eps$-order terms for both sides we obtain \eqref{eq:phib:uk}.

  To obtain \eqref{eq:phib:uij}, note that from \eqref{eq:phi:ui}, for all $i \neq j$, $u_i \in \cU_i$ and $u_j \in \cU_j$, we have
  \begin{subequations}
    \begin{align}
        \eps^2\left\langle\phib_{u_i}, \phib_{u_j}\right\rangle%
      &= \eps^2\sum_{x^d} \phi_{u_i}(x^d)\phi_{u_j}(x^d)\\
      &= \sum_{x_1^d}\biggl(\frac{1}{P_{X^d}(x^d)}\cdot [P_{X^d|U_i}(x^d|u_i) - P_{X^d}(x^d)]%
        \cdot[P_{X^d|U_j}(x^d|u_i) - P_{X^d}(x^d)]\biggr)\\
      &= \sum_{x_1^d}\frac{P_{X^d|U_i}(x^d|u_i)P_{X^d|U_j}(x^d|u_j)}{P_{X^d}(x^d)} - 1\\
      &= \sum_{x_1^d} P_{X^d}(x^d)\cdot\frac{P_{U_i|X^d}(u_i|x^d)}{P_{U_i}(u_i)}\cdot\frac{P_{U_j|X^d}(u_j|x^d)}{P_{U_j}(u_j)} - 1\\
      &= \frac{1}{P_{U_iU_j}(u_i, u_j)}\sum_{x_1^d} P_{X^d}(x^d)P_{U_iU_j|X^d}(u_i,u_j|x^d) - 1\label{eq:uij}\\
      &= 0,
    \end{align}
  \end{subequations}
  where to obtain \eqref{eq:uij} we have again exploited the independence and conditional independence of $U_i$ and $U_j$.   
  
  In addition, the objective function $\cL(X^d|U^k)$ can be expressed as [cf. \eqref{eq:CI}]
  \begin{align}
      D(P_{X^d}\|P_{X_1}\dots P_{X_d}) - D(P_{X^d}\|P_{X_1}\dots P_{X_d}|U^k)\notag%
    &= \sum_{i = 1}^d I(U^k;X_i) - I(U^k;X^d)\notag\\
    &= \sum_{i = 1}^d I(U^k;X_i) - \sum_{j = 1}^kI(U_j;X^d),
  \end{align}
  where to obtain the last equality we have used the fact that
  \begin{subequations}
    \begin{align}
      I(U^k;X^d)
      &= \E_{U^kX^d}\left[\log \frac{P_{U^k|X^d}(U^k|X^d)}{P_{U^k}(U^k)}\right]\\
      &= \E_{U^kX^d}\left[\sum_{j = 1}^k\log \frac{P_{U_j|X^d}(U_j|X^d)}{P_{U_j}(U_j)}\right]\label{eq:u:x}\\
      &= \sum_{j = 1}^k I(U_j; X^d),
    \end{align}
  \end{subequations}
  and where \eqref{eq:u:x} follows from the facts that $U_1, \dots, U_k$ are mutually independent and are conditionally independent given $X^d$.

  For each $i$, the mutual information $I(U^k; X_i)$ can be approximated as
  \begin{align*}
    I(U^k; X_i) = \frac{1}{2} \eps^2 \E_{U^k}\left[\left\|\psib_{i, U^k}\right\|^2\right] + o(\eps^2), 
  \end{align*}
  where for $U^k = u^k$, the vector $\psib_{i, u^k}$ is an $|\cX_i|$-dimensional perturbation vector defined as
  \begin{align*}
    \psi_{i, u^k}(x_i) = \frac{P_{X_i|U^k}(x_i|u^k) - P_{X_i}(x_i)}{\eps\sqrt{P_{X_i}(x_i)}}.
  \end{align*}
  Therefore, by ignoring the higher order terms of $\eps$, the maximization of total correlation can be rewritten as
  \begin{subequations}
    \begin{alignat}{2}
      &\quad\max_{\phib_{u^k}} & ~& \sum_{i = 1}^d \E_{U^k}\left[\left\|\psib_{i, U^k}\right\|^2\right] - \sum_{j = 1}^k \E_{U_j}\left[\left\|\phib_{U_j}\right\|^2\right]    \label{eq:opt:uk:obj}\\
      &\text{subject to:}& & 1 \geq \E_{U_1}\left[\left\|\phib_{U_1}\right\|^2\right] \geq \dots \geq \E_{U_k}\left[\left\|\phib_{U_k}\right\|^2\right],%
      \label{eq:opt:uk:len}\\
      & & & \langle\phib_{u_i}, \phib_{u_j}\rangle = 0, \quad i \neq j, u_i \in \cU_i, u_j \in \cU_j,\label{eq:opt:uk:ortho}\\
      & & & \left\langle\phib_{u_j}, \phib^{(0)}\right\rangle = 0, \forall\, u_j \in \cU_j, j = 1, \dots, k, \label{eq:opt:uk:ortho:0}\\
      & & & \phib_{u^k} = \sum_{j = 1}^k \phib_{u_j}, \forall\, u^k \in \cU_1 \times \cdots \times \cU_{k}. \label{eq:opt:uk:decomp}
    \end{alignat}
    \label{eq:opt:uk}
\end{subequations}

To solve \eqref{eq:opt:uk}, first observe that we have $\psib_{i, U^k} = \dtm_i \phib_{U^k}$, where $\dtm_i$ is as defined in \eqref{eq:def:Bi}. Then, the objective function \eqref{eq:opt:uk:obj} can be rewritten as
\begin{subequations}
  \begin{align}
      \sum_{i = 1}^d\E_{U^k}\left[\left\|\psib_{i, U^k}\right\|^2\right] - \sum_{j = 1}^k \E_{U_j}\left[\left\|\phib_{U_j}\right\|^2\right]%
    &= \sum_{i = 1}^d\E_{U^k}\left[\left\|\dtm_i\phib_{U^k}\right\|^2\right] - \sum_{j = 1}^k \E_{U_j}\left[\left\|\phib_{U_j}\right\|^2\right]\label{eq:opt:uk:obj:1} \\
    &= \E_{U^k}\left[\left\|\dtm_0\phib_{U^k}\right\|^2\right] - \sum_{j = 1}^k \E_{U_j}\left[\left\|\phib_{U_j}\right\|^2\right]\label{eq:opt:uk:obj:2}\\
    &= \E_{U^k}\left[\left\|\sum_{j = 1}^k\dtm_0\phib_{U_j}\right\|^2\right] - \sum_{j = 1}^k \E_{U_j}\left[\left\|\phib_{U_j}\right\|^2\right]\label{eq:opt:uk:obj:3}\\
    &= \sum_{j = 1}^k \E_{U_j}\left[\left\|\dtm_0\phib_{U_j}\right\|^2\right] - \sum_{j = 1}^k \E_{U_j}\left[\left\|\phib_{U_j}\right\|^2\right]\label{eq:opt:uk:obj:4}\\
    &= \sum_{j = 1}^k \E_{U_j}\left[\left\|\dtm_0\phib_{U_j}\right\|^2 - \left\|\phib_{U_j}\right\|^2\right]\label{eq:opt:uk:obj:5}%
  \end{align}
\end{subequations}
where $\dtm_0$ is as defined in \eqref{eq:B0}. To obtain \eqref{eq:opt:uk:obj:4}, we have used the fact that
\begin{align*}
  \E_{U^k}\left[\phib_{U_i}^{\T}\dtm_0^{\T}\dtm_0\phib_{U_j}\right]
  &= \left(\E_{U_i}\left[\phib_{U_i}\right]\right)^{\T}\dtm_0^{\T}\dtm_0\left(\E_{U_j}\left[\phib_{U_j}\right]\right)%
    = 0, \quad i \neq j,
\end{align*}
where the first equality follows from the fact that $U_i$ and $U_j$ are independent, and the second equality follows from that $\E_{U_i}\left[\phib_{U_i}\right] = 0$.

To maximize \eqref{eq:opt:uk:obj:5}, %
$\phib_{u_i}$ should be aligned to the same direction for all $u_i \in \cU_i$. Otherwise, we can align all $\phib_{u_i}$ to 
\begin{align*}
  \argmax_{\phib_{u_i}\colon u_i \in \cU_i} \frac{\|\dtm_0\phib_{u_i}\|^2}{\|\phib_{u_i}\|^2}
\end{align*}
while keeping $\E_{U_i}[\|\phib_{U_i}\|^2]$ fixed, which yields a larger value for the objective function.

Therefore, for each $i$ and $u_i \in \cU_i$, we can write $\phib_{u_i}$ as
\begin{align}
  \phib_{u_i} = h_i(u_i)\phib_i,
  \label{eq:h_i}
\end{align}
where $h_i \colon \cU_i \mapsto \mathbb{R}$ and $\phib_i$ is a unit-norm vector. Then, we have
  $\E_{U_i}[\phib_{U_i}] = \E_{U_i}[h_i(U_i)] \phib_i = 0$
and
\begin{subequations}
  \begin{gather}
    \E_{U_i}[\|\phib_{U_i}\|^2] = \E_{U_i}[h^2_i(U_i)],\\
    \E_{U_i}[\|\dtm_0\phib_{U_i}\|^2] = \E_{U_i}[h^2_i(U_i)] \|\dtm_0\phib_i\|^2.
  \end{gather}
  \label{eq:E:ui}
\end{subequations}
Now, the constraint \eqref{eq:opt:uk:len} can be reduced to
\begin{align}
  1 \geq \E_{U_1}[h^2_1(U_1)] \geq \dots \geq \E_{U_k}[h^2_k(U_k)].
  \label{eq:ui:order}
\end{align}
In addition, it follows from \eqref{eq:E:ui} that
    $\E_{U_i}[\|\dtm_0\phib_{U_i}\|^2] -  \E_{U_i}[\|\phib_{U_i}\|^2]  %
    = \E_{U_i}[h^2_i(U_i)] \left[\|\dtm_0\phib_i\|^2 - 1\right]$.
As a result, to maximize \eqref{eq:opt:uk:obj:5}, $h_i$ should be chosen such that
\begin{align*}
  \E_{U_i}[h^2_i(U_i)] =
  \begin{cases}
    1& \text{if~}\|\dtm_0\phib_i\|^2 > 1,\\
    0& \text{otherwise.}\\
  \end{cases}
\end{align*}

Then, from \eqref{eq:ui:order} there exists $k_0 \in \{1, \dots, k\}$ such that
\begin{align}
  \E_{U_i}[h^2_i(U_i)] =
  \begin{cases}
    1& i = 1, \dots, k_0,\\
    0& i > k_0,\\
  \end{cases}
  \label{eq:e:h^2}
\end{align}
and the objective function \eqref{eq:opt:uk:obj:5} can be reduced to
\begin{align*}
  \sum_{j = 1}^k \E_{U_j}\left[\left\|\dtm_0\phib_{U_j}\right\|^2 - \left\|\phib_{U_j}\right\|^2\right]
  &= \sum_{j = 1}^{k_0} \|\dtm_0\phib_i\|^2 - k_0%
    = \bbfrob{\dtm_0\Phib_0}^2 - k_0,
\end{align*}
where we have defined  $\Phib_0 \defeq
\begin{bmatrix}
  \phib_1& \cdots & \phib_{k_0}
\end{bmatrix}
$.

As a result, the optimization problem \eqref{eq:opt:uk} is equivalent to
\begin{subequations}
  \begin{alignat}{2}
    &\quad\max_{\Phib_{0}} & ~ & \bbfrob{\dtm_0\Phib_0}^2 - k_0 \label{eq:opt:Phib:1}\\
    &\text{subject to:} & & \Phib_{0}^{\T}\Phib_{0} = \Ib_{k_0},   \label{eq:opt:Phib:2}\\
    & & & \Phib_{0}^{\T}\phib^{(0)} = \zerob_{k_0},   \label{eq:opt:Phib:3}
  \end{alignat}
  \label{eq:opt:Phib}
\end{subequations}
where $\Ib_{k_0}$ is the identity matrix of order $k_0$, and $\zerob_{k_0}$ is the zero vector in $\mathbb{R}^{k_0}$. In addition, since $\phib^{(0)}$ is the first right singular vector of $\dtm_0$, \eqref{eq:opt:Phib} can be further reduced to 
\begin{subequations}
  \begin{alignat}{2}
    &\quad\max_{\Phib_{0}} & ~ &\bbfrob{\dtmt_0\Phib_0}^2 - k_0 \label{eq:opt:Phib:new:1}\\
    &\text{subject to:} & &\Phib_{0}^{\T}\Phib_{0} = \Ib_{k_0},   \label{eq:opt:Phib:new:2}
  \end{alignat}
  \label{eq:opt:Phib:new}
\end{subequations}
where $\dtmt_0 \defeq \dtm_0 - \sqrt{\lambda^{(0)}}\psib^{0}\left(\phib^{(0)}\right)^{\T}$.

From Lemma \ref{lem:frob}, the optimal value of \eqref{eq:opt:Phib:new} is
\begin{align}
  \sum_{i = 1}^{k_0} \lambda^{(i)} - k_0 = \sum_{i = 1}^{k_0} \left[\lambda^{(i)} - 1\right].\label{eq:val:k0}
\end{align}
To maximize \eqref{eq:val:k0}, $k_0$ should be chosen as the largest $i$ such that $\lambda^{(i)} > 1$, i.e., $k_0 = \min\{k, k^*\}$. In addition, the optimal $\Phib_0$ is
  $\Phib_0 =
  \begin{bmatrix}
    \phib^{(1)} & \cdots & \phib^{(k_0)} 
  \end{bmatrix}
  \Qb$
for $\Qb \in \mathbb{R}^{k_0 \times k_0}$ with $\Qb^{\T}\Qb = \Ib_{k_0}$. Hence, we have
\begin{align*}
  \phib_{\ell} = \sum_{j = 1}^{k_0} q_{j\ell}\phib^{(j)}
\end{align*}

Following the same derivation as that for \eqref{eq:phi:1}, we can express $\phib^{(j)}$ as
\begin{align*}
  \frac{\phi^{(j)}(x^d)}{\sqrt{P_{X^d}(x^d)}} = \frac{1}{\sqrt{\lambda^{(j)}}}\sum_{i = 1}^d f_i^{(j)}(x_i).
\end{align*}
and thus
\begin{align*}
  \frac{\phi_{\ell}(x^d)}{\sqrt{P_{X^d}(x^d)}}
  &= \sum_{j = 1}^{k_0} q_{j\ell} \cdot \frac{\phi^{(j)}(x^d)}{\sqrt{P_{X^d}(x^d)}}%
    = \sum_{j = 1}^{k_0} \frac{q_{j\ell}}{\sqrt{\lambda^{(j)}}} \sum_{i = 1}^d f_i^{(j)}(x_i).
\end{align*}

Then, it follows from \eqref{eq:h_i} that
\begin{align} %
  \frac{\phi_{u_\ell}(x^d)}{\sqrt{P_{X^d}(x^d)}} = h_\ell(u_{\ell})\sum_{j = 1}^{k_0} \frac{q_{j\ell}}{\sqrt{\lambda^{(j)}}} \sum_{i = 1}^d f_i^{(j)}(x_i)
\end{align}
for $\ell = 1, \dots, k_0$. Moreover, from \eqref{eq:opt:uk:decomp}, we have
\begin{align*}
  {\phib_{u^{k}}}
  = \sum_{\ell = 1}^k \phib_{u_{\ell}}
  = \sum_{\ell = 1}^{k_0} \phib_{u_{\ell}},
\end{align*}
where the second equality follows from the consequence of \eqref{eq:h_i} and \eqref{eq:e:h^2} that $\phib_{u_{\ell}} = \zerob$ for $\ell > k_0$.

Therefore,
\begin{align*} %
  \frac{\phi_{u^{k}}(x^d)}{\sqrt{P_{X^d}(x^d)}}
  &= \sum_{\ell = 1}^{k_0}\frac{\phi_{u_\ell}(x^d)}{\sqrt{P_{X^d}(x^d)}}%
    = \sum_{\ell = 1}^{k_0} h_\ell(u_{\ell})\sum_{j = 1}^{k_0} \frac{q_{j\ell}}{\sqrt{\lambda^{(j)}}} \sum_{i = 1}^d f_i^{(j)}(x_i),
\end{align*}
which implies
\begin{align*} %
  P_{X^d|U^{k}} (x^d |u^{k}) %
&= P_{X^d} (x^d)  \left[ 1+ \epsilon   \frac{\phi_{u^{k}}(x^d)}{\sqrt{P_{X^d}(x^d)}} \right] + o(\epsilon)\\
&= P_{X^d} (x^d)  \left[ 1 + \epsilon \sum_{\ell = 1}^{k_0} h_\ell(u_{\ell})\sum_{j = 1}^{k_0} \frac{q_{j\ell}}{\sqrt{\lambda^{(j)}}} \sum_{i = 1}^d f_i^{(j)}(x_i)
  \right] + o(\epsilon)
\end{align*}
and 
\begin{align} %
&P_{X^dU^{k}} (x^d,u^{k}) = P_{X^d} (x^d) \left[\prod_{j = 1}^{k} P_{U_j}(u_j)\right] %
  \cdot\left[ 1+ \epsilon \sum_{\ell = 1}^{k_0} h_\ell(u_{\ell})\sum_{j = 1}^{k_0} \frac{q_{j\ell}}{\sqrt{\lambda^{(j)}}} \sum_{i = 1}^d f_i^{(j)}(x_i) \right] + o(\epsilon).\label{eq:P:xd:uk}
\end{align}

Finally, note that the exponential family ${\cal P}^{(\delta)}_{\exp, k}$, when $\delta$ is small, can be written as
\begin{align} 
    {\cal P}^{(\delta)}_{\exp, k}
  = &\left\{ P_{X^d} (x^d) \left[\prod_{j = 1}^{k} P_{U_j}(u_j) \right]\right.  %
    \cdot\left[ 1+ \sqrt{2\delta}  \sum_{\ell = 1}^{k_0} h_\ell(u_{\ell})\sum_{j = 1}^{k_0} \frac{q_{j\ell}}{\sqrt{\lambda^{(j)}}} \sum_{i = 1}^d f_i^{(j)}(x_i)\right]\notag\\
  &\qquad: \left.  h_\ell \in {\cH}_\ell, \Qb = [q_{ij}]_{k_0 \times k_0}, \Qb^{\T}\Qb = \Ib_{k_0} \vphantom{\sum_{j = 1}^{k_0}}\right\}.
    \label{eq:exp:k}
\end{align}
Since $\delta = \frac{1}{2} \epsilon^2$ the proof is completed by comparing~\eqref{eq:exp:k} and~\eqref{eq:P:xd:uk}. 

\section{Joint Correlation Maximization}
\label{sec:a:mace}
For functions $\uf_i\colon \cX_i \mapsto \mathbb{R}^k$, $i = 1, \dots, d$, we define $\Psib_i \in \mathbb{R}^{|\cX_i| \times k}$ such that the row vectors of $\Psib_i$ are $\sqrt{P_{X_i}(x_i)}\uf_i^{\T}(x_i)$, for all $x_i \in \cX_i$. Furthermore, we define the $m \times k$ matrix $\Psib$ as $\Psib =
\begin{bmatrix}
\Psib_1^{\T} & \cdots & \Psib_d
\end{bmatrix}.
$
Then the optimization problem \eqref{eq:MHGRk} can be rewritten as
\begin{subequations}
  \begin{alignat}{2} 
    &\max_{\Psib : \Psib \in \mathbb{R}^{m \times k}}  & ~ &\trop{\Psib^\T\dtm\Psib} \\ %
    &~\text{subject to:} \ & &\Psib_i^{\T}\vb_i = \zerob_k, \  \text{for all $i$}, \\
    & & &  \Psib^{\T}\Psib = \Ib_k,
  \end{alignat}
  \label{eq:MHGRk:equiv}
\end{subequations}
where $\zerob_k$ is the zero vector in $\mathbb{R}^k$, and $\Ib_k$ is the $k \times k$ identity matrix. To establish the equivalence of \eqref{eq:MHGRk} and \eqref{eq:MHGRk:equiv}, note that we have
\begin{align*}
  \Psib^{\T}\Psib
  = \sum_{i = 1}^d \Psib_i^{\T}\Psib_i
  &= \sum_{i = 1}^d \sum_{x_i \in \cX_i} P_{X_i}(x_i)\uf_i(x_i)\uf^{\T}_i(x_i)\\
  &= \sum_{i = 1}^d \E\left[\uf_i(X_i)\uf^{\T}_i(X_i) \right]%
    = \E\left[\sum_{i = 1}^d \uf_i(X_i)\uf^{\T}_i(X_i) \right]
\end{align*}
and
\begin{align*}
  \trop{\Psib^\T\dtm\Psib}
  &= \sum_{i = 1}^d\sum_{j = 1}^d \trop{\Psib_i^\T \dtm_{ij} \Psib_j}%
    = \sum_{i = 1}^d\sum_{j = 1}^d \trop{\E\left[\uf_i(X_i)\uf_j^{\T}(X_j)\right]}\\
  &= \sum_{i = 1}^d \trop{\E\left[\uf_i(X_i)\uf_i^{\T}(X_i)\right]}%
    + \sum_{i \neq j} \trop{\E\left[\uf_i(X_i)\uf_j^{\T}(X_j)\right]}\\
  &= \trop{\sum_{i = 1}^d \E\left[\uf_i(X_i)\uf_i^{\T}(X_i)\right]}%
    + \sum_{i \neq j} \trop{\E\left[\uf_i(X_i)\uf_j^{\T}(X_j)\right]}\\
  &= k + \E\left[\sum_{i \neq j} \uf^{\T}_i(X_i)\uf_j(X_j)\right].
\end{align*}

From Lemma \ref{lem:1}, for $k < m - d$, the solution of \eqref{eq:MHGRk:equiv} can be represented as
  $\Psib^* =
  \begin{bmatrix}
    \psib^{(1)} & \cdots & \psib^{(k)}
  \end{bmatrix}
  \Qb$, 
where $\Qb \in \mathbb{R}^{k \times k}$ is an orthogonal matrix. Therefore, the optimal solution of \eqref{eq:MHGRk} corresponds to $f_i^{(\ell)}$ with $i = 1, \dots, d$ and $\ell = 1, \dots, k$.%

\section{Common Bits Patterns Extraction}%
\label{sec:a:prop:bits}

First, we define $\ell_{\max}$ as the largest $\ell$ such that $w(\cJ_{\ell}) > 0$, i.e., $\ell_{\max} \defeq \max\{\ell\colon 0 \leq \ell \leq 2^r - 1, w(\cJ_{\ell}) > 0\}$. Then, $w(\cJ_\ell) > 0$ is equivalent to $\ell \leq \ell_{\max}$, and \eqref{eq:bits:lambda} can be equivalently expressed as
  \begin{align}
    \lambda^{(\ell)} = w(\cJ_{\ell}), \quad \ell \leq \ell_{\max},
    \label{eq:lambda:pos}
  \end{align}
  and
  \begin{align}
    \lambda^{(\ell)} = 0, \quad \ell > \ell_{\max}.
    \label{eq:lambda:0}
  \end{align}
    
  Note that \eqref{eq:fpsi} establishes a one-to-one correspondence between the functions $f_i^{(\ell)}$ ($i = 1, \dots, d$) and the vector $\psib^{(\ell)}$. With this correspondence, we use $\psibt^{(\ell)}$ to denote the vector corresponding to the functions $f_i^{(\ell)}$ as defined in \eqref{eq:f:bits}. Then the proof can be accomplished in two steps. First, we show that $\psibt^{(\ell)}$ ($\ell = 0, \dots, \ell_{\max}$) are $(\ell_{\max} + 1)$ orthogonal eigenvectors of $\dtm$ associated with eigenvalues $w(\cJ_{\ell})$ ($\ell = 0, \dots, \ell_{\max}$), i.e., for all $0 \leq \ell \leq \ell_{\max}$ and $0 \leq \ell' \leq \ell_{\max}$, the $\psibt^{(\ell)}$'s satisfy
  \begin{align}
    \dtm\psibt^{(\ell)} = w(\cJ_{\ell})\psibt^{(\ell)} \quad\text{and}\quad \left\langle\psibt^{(\ell)}, \psibt^{(\ell')}\right\rangle = \delta_{\ell\ell'},
    \label{eq:eig:psibt}
  \end{align}
  where $\delta_{\ell\ell'}$ is the Kronecker delta. Then, it suffices to verify that all other eigenvalues of $\dtm$ are zeros [cf. \eqref{eq:lambda:0}]. 

  To begin, we equivalently express \eqref{eq:eig:psibt} using $f_i^{(\ell)}$ as
  \begin{align}
    \sum_{j = 1}^{d}\E\left[f^{(\ell)}_j(X_j)\middle|X_i\right] = w(\cJ_{\ell}) f^{(\ell)}_i(X_i), \quad 1 \leq i \leq d,
    \label{eq:eig:f:1}
  \end{align}
  and 
  \begin{align}
    \sum_{i = 1}^d\E\left[f^{(\ell)}_i(X_i)f^{(\ell')}_i(X_i)\right] = \delta_{\ell\ell'}.
    \label{eq:eig:f:2}
  \end{align}
  Then, since we have [cf. \eqref{eq:def:w}]
  \begin{align*}
    \sum_{i = 1}^d \1_{\{\cJ_\ell \subset \cI_i\}} = \sum_{j = 1}^d \1_{\{\cJ_\ell \subset \cI_j\}} = w(\cJ_{\ell}),
  \end{align*}
  it suffices to show that
  \begin{align}
    \E\left[f^{(\ell)}_j(X_j)\middle|X_i\right] = f^{(\ell)}_i(X_i) \cdot \1_{\{\cJ_\ell \subset \cI_j\}}, \quad 1 \leq i, j \leq d,
    \label{eq:eig:f:1:new}
  \end{align}
  and
  \begin{align}
    \E\left[f^{(\ell)}_i(X_i)f^{(\ell')}_i(X_i)\right] = \frac{\1_{\{\cJ_{\ell} \subset \cI_i\}}}{w(\cJ_{\ell})} \cdot \delta_{\ell\ell'},\quad 1 \leq i \leq d.
    \label{eq:eig:f:2:new}
  \end{align}

  To obtain \eqref{eq:eig:f:1:new}, note that if $\cJ_{\ell} \not\subset \cI_j$, it follows from \eqref{eq:f:bits} that $f_j(X_j) = 0$, and thus \eqref{eq:eig:f:1:new} holds. Otherwise, we have $\cJ_{\ell} \subset \cI_j$ and 
  \begin{align}
    \E\left[f^{(\ell)}_j(X_j)\middle|X_i\right]
    &=  \frac{1}{\sqrt{w(\cJ_\ell)}}\E\left[\prod_{s \in \cJ_{\ell}} b_s\middle|X_i\right].
    \label{eq:f:ce:new}
  \end{align}
  Since $X_i = b_{\cI_i}$ is composed of all the $b_s$'s with indices in $\cI_i$, %
  we have
  \begin{align*}
    \E\left[\prod_{s \in \cJ_{\ell}} b_s\middle|X_i\right] =
    \begin{cases}
      \displaystyle\prod_{s \in \cJ_{\ell}} b_s,& \text{if}~ \cJ_{\ell} \subset \cI_i, \\
      \displaystyle 0,& \text{otherwise}.
    \end{cases}
  \end{align*}
  Therefore, we obtain
  \begin{align*}
    \E\left[f^{(\ell)}_j(X_j)\middle|X_i\right]
     =\frac{1}{\sqrt{w(\cJ_\ell)}}\E\left[\prod_{s \in \cJ_{\ell}} b_s\middle|X_i\right]%
    &=
    \begin{cases}
      \displaystyle\frac{1}{\sqrt{w(\cJ_\ell)}}\prod_{s \in \cJ_{\ell}} b_s,& \text{if}~ \cJ_{\ell} \subset \cI_i, \\
      \displaystyle 0,& \text{otherwise}
    \end{cases}\\
    &= f^{(\ell)}_i(X_i)%
      = f^{(\ell)}_i(X_i) \cdot \1_{\{\cJ_{\ell} \subset \cI_j\}}.
  \end{align*}

  Likewise, \eqref{eq:eig:f:2:new} follows immediately from \eqref{eq:f:bits} when $\ell = \ell'$, and it suffices to consider the case $\ell \neq \ell'$ and prove that
  \begin{align}
    \E\left[f^{(\ell)}_i(X_i)f^{(\ell')}_i(X_i)\right] = 0.
    \label{eq:f:ell:ell'}
  \end{align}
  Indeed, when $\cJ_{\ell} \not\subset \cI_i$ or $\cJ_{\ell'} \not\subset \cI_i$, \eqref{eq:f:ell:ell'} is trivially true. Otherwise, we have $\cJ_{\ell} \subset \cI_i$ and $\cJ_{\ell'} \subset \cI_i$, and it follows from \eqref{eq:f:bits} that
  \begin{align*}
    f^{(\ell)}_i(X_i)f^{(\ell')}_i(X_i) = \frac{1}{\sqrt{w(\cJ_{\ell})w(\cJ_{\ell'})}}\prod_{j \in \cJ_{\ell} \symdif \cJ_{\ell'}} b_j,
  \end{align*}
  where ``$\symdif$'' denotes the symmetric difference of two sets, i.e., $A \symdif B = \left(A\setminus B\right)\cup \left(B\setminus A\right)$. Therefore, we have
  \begin{align*}
    \E\left[f^{(\ell)}_i(X_i)f^{(\ell')}_i(X_i)\right]
    &= \frac{1}{\sqrt{w(\cJ_{\ell})w(\cJ_{\ell'})}}\prod_{j \in \cJ_{\ell} \symdif \cJ_{\ell'}} \E[b_j]%
      = 0,
  \end{align*}
  where we have used the fact that the set $(\cJ_{\ell} \symdif \cJ_{\ell'})$ is non-empty, since $\cJ_{\ell} \neq \cJ_{\ell'}$.
  
  Finally, to prove \eqref{eq:lambda:0}, i.e., eigenvalues other that $w(\cJ_{\ell})$ ($\ell = 0, \dots, \ell_{\max}$) are all zeros, note that
  \begin{align*}
    \sum_{\ell = 0}^{\ell_{\max}} w(\cJ_{\ell})
    = \sum_{\ell = 0}^{2^{r} - 1} w(\cJ_{\ell})
    &
      = \sum_{\cI \subset [r]} w(\cI)\\
    &
      = \sum_{\cI \subset [r]}\sum_{i = 1}^d \1_{\{\cI \subset \cI_i\}}\\
    &= \sum_{i = 1}^d \sum_{\cI \subset [r]} \1_{\{\cI \subset \cI_i\}} \\
    &= \sum_{i = 1}^d 2^{|\cI_i|}%
      = \sum_{i = 1}^d |\cX_i| = m.
  \end{align*}
  On the other hand, we have the sum of all eigenvalues
  \begin{align*}
    \sum_{\ell = 0}^{m-1} \lambda^{(\ell)} = \trop{\dtm} = m.
  \end{align*}
  From Lemma \ref{lem:1}, all eigenvalues of $\dtm$ are non-negative, which implies \eqref{eq:lambda:0}.

  \section{Proof of Proposition \ref{prop:MH-score}}
  \label{sec:a:MH}
To begin, we write the matrix $\dtmt$ of \eqref{eq:dtmt} as a block matrix
\begin{align} \label{eq:Bt}
  \dtmt = 
  \left[
  \begin{array}{cccc}
    \dtmt_{11} & \dtmt_{12} & \cdots & \dtmt_{1d} \\
    \dtmt_{21} & \dtmt_{22} & \cdots & \dtmt_{2d} \\
    \vdots & \vdots & \ddots & \vdots \\
    \dtmt_{d1} & \dtmt_{d2}  & \cdots & \dtmt_{dd}
  \end{array} 
                                        \right],
\end{align}
where each block $\dtmt_{ij}$ is an $( | \cX_i | \times | \cX_j | )$-dimensional matrix. Then, we can rewrite $\bfrob{\dtmt - \Psib\Psib^{\T}}^2$ as
\begin{align}
    \bfrob{\dtmt - \Psib\Psib^{\T}}^2%
  &= \sum_{i = 1}^d\sum_{j = 1}^d\bfrob{\dtmt_{ij} - \Psib_i\Psib_j^{\T}}^2\notag\\
  &= \sum_{i = 1}^d\sum_{j = 1}^d \left[\bfrob{\dtmt_{ij}}^2 - 2\trop{\Psib_i^\T \dtmt_{ij}\Psib_j} + \bfrob{\Psib_i\Psib_j}^2\right]\notag\\
  &= \sum_{i = 1}^d\sum_{j = 1}^d \left[\bfrob{\dtmt_{ij}}^2 - 2H\left(\uf_i(X_i), \uf_j(X_j)\right)\right]\notag\\
  &= \bfrob{\dtmt}^2 - 2H\left(\uf_1(X_1), \dots, \uf_d(X_d)\right),\label{eq:1}
\end{align}
where we have used the fact that
\begin{align*}
    \trop{\Psib_i^\T \dtmt_{ij}\Psib_j} - \frac{1}{2} \bfrob{\Psib_i\Psib_j}^2%
  &= \E\left[\uf^{\T}_i(X_i)\uf_j(X_j)\right] - \left(\E\left[\uf_i(X_i)\right]\right)^\T \E\left[\uf_j(X_j)\right]\\
  &\quad
    - \frac12 \trop{\E\left[\uf_i(X_i)\uf^\T_i(X_i)\right]\E\left[\uf_j(X_i)\uf^\T_j(X_j)\right]}\\
  &= H\left(\uf_i(X_i), \uf_j(X_j)\right).
\end{align*}

\if\secsc1
\section{Proof of Theorem \ref{thm:exponent}}
\label{sec:a:thm:exponent}
  First, we show in the following that \eqref{eq:opt} is equivalent to the optimization problem without the equality constraint, i.e,
  \begin{align}
    &\quad\max\quad ~ \sum_{i = 1}^k\sum_{j = k + 1}^{m} \frac{\left[\left(\psib^{(i)}\right)^{\T}
      \dtmdmat \psib^{(j)}\right]^2}{\lambda^{(i)} - \lambda^{(j)}}\notag\\
    &\text{subject to:} ~ \|\xib_{X^d}\|^2 \leq 1.
      \label{eq:opt:equiv}
  \end{align}
  To see this, suppose that $\xib_{X^d}^*$ is the optimal solution of \eqref{eq:opt:equiv} with $c \defeq \sum_{x^d} \sqrt{P_{X^d}(x^d)} \xi_{X^d}^*$. Let $z(x^d) = \xi_{X^d}^* (x^d) - c\sqrt{P_{X^d}(x^d)}$, then we have
  \begin{align*}
    1 = \sum_{x^d}\left[\xi^*_{X^d}(x^d)\right]^2 = \sum_{x^d}z^2(x^d) + c^2,
  \end{align*}
  which implies $|c| \leq 1$. 

  If $|c| = 1$, we have $\xi^*_{X^d}(x^d) = \pm \sqrt{P_{X^d}(x^d)}$, and it follows from \eqref{eq:Xi}--\eqref{eq:Xi:ij} that $\Xib = \mp d \psib^{(0)}\left(\psib^{(0)}\right)^{\T}$, which implies that the objective function, and thus contradicts the assumption that $\xib^*_{X^d}$ is optimal. On the other hand, if $0 < |c| < 1$, then we can construct the vector $\xib_{X^d}'$ with entries $\xi_{X^d}(x^d) = z(x^d)/\sqrt{1 - c^2}$. Under this construction, $\|\xib_{X^d}'\|^2 = 1$, and the objective function in \eqref{eq:opt} for $\xib_{X^d}'$ is $1/(1 - c^2)$ times the corresponding value for $\xib_{X^d}^*$. This again contradicts the optimality of $\xib_{X^d}^*$. Therefore, we have $c = 0$, and the optimization problem \eqref{eq:opt:equiv} has the same solution as \eqref{eq:opt}. 

  Next, to reduce the objective function of \eqref{eq:opt:equiv}, note that $\Xib$ can be represented as a linear function of all $\xi_{X_i}$ ($i = 1, \dots, d$) and $\xi_{X_iX_j}$ ($i \neq j$), and thus is linear to the vector $\xib_{X^d}$. Therefore, the objective function of \eqref{eq:opt:equiv} can be expressed as a quadratic function of $\xib_{X^d}$. For this purpose, we first generalize the definition \eqref{eq:xi:xixj} to the case $i = j$ by defining
  \begin{align}
    \xi_{X_iX_i}(x_i, x_i') \defeq \xi_{X_i}(x_i)\delta_{x_ix_i'}
    \label{eq:xi:XiXi:def}    
  \end{align}
 for all $i = 1, \dots, d$ and $x_i, x_i' \in \cX_i$. In addition, for all $(i, j)$, we define $\xib_{X_iX_j}$ as an $(|\cX_i|\cdot|\cX_j|)$-dimensional vector whose $[(x_j' - 1)|\cX_i| + x_i]$-th entry is $\xi_{X_iX_j}(x_i, x_j')$. Then, it follows from \eqref{eq:Xi:ii} and \eqref{eq:Xi:ij} that\footnote{The vectorization operation $\vecop(\cdot)$ stacks all columns of a matrix into a vector. Specifically, for $\Wb = [w_{ij}]_{p\times q}$, $\vecop(\Wb)$ is a $(pq)$-dimensional column vector with $[p(j-1) + i]$-th entry being $w_{ij}$.}
  \begin{align}
    \vecop(\Xib_{ij}) = \Lb_{ij}\xib_{X_iX_j},\quad\text{for~all~}i, j,
    \label{eq:Xib:ij:xib:pair}
  \end{align}
  with $\Lb_{ij}$ as defined in \eqref{eq:Lb:ij}.

  Furthermore, from \eqref{eq:xi:x^d} and \eqref{eq:xi:xixj}--\eqref{eq:xi}, for all $(i, j)$ we have
  \begin{align}
    \xib_{X_iX_j} = \Cb_{ij}\xib_{X^d},
    \label{eq:xib:2-d}    
  \end{align}
 where $\Cb_{ij}$ is an $(|\cX_i|\cdot|\cX_j|) \times (|\cX_1|\cdot|\cX_2|\cdots|\cX_d|)$ matrix with entries
  \begin{align}
    &C_{ij}\left((x'_i, \hat{x}_j); x^d\right) \notag\\
    &=  
      \begin{cases}
        \displaystyle\frac{\sqrt{P_{X^d}(x^d)}}{\sqrt{P_{X_iX_j}(x'_i, \hat{x}_j)}}\delta_{x_ix'_i}\delta_{x_j\hat{x}_j},&\text{if}~P_{X_iX_j}(x'_i, \hat{x}_j) > 0,\\
        0,&\text{otherwise}.
      \end{cases}
            \label{eq:Cij:def}
  \end{align}
 
  Moreover, we define an $m^2$-dimensional vector $\zetab$ by stacking all $\vecop(\Xib_{ij})$'s $(1 \leq i \leq d, 1 \leq j \leq d)$ to a column vector, i.e., 
  \begin{align}
    \zetab \defeq
    \begin{bmatrix}
      \vecop(\Xib_{11})\\
      \vecop(\Xib_{21})\\
      \vdots\\
      \vecop(\Xib_{dd})
    \end{bmatrix}.
    \label{eq:def:zetab}
  \end{align}
  Then, it follows from \eqref{eq:Xib:ij:xib:pair}--\eqref{eq:def:zetab} that  
  \begin{align}
    \zetab = \Jb_0 \xib_{X^d},
    \label{eq:zetab:xib}
  \end{align}
  where we have defined $\Jb_0 \in \mathbb{R}^{m^2 \times (|\cX_1|\cdot|\cX_2|\cdots|\cX_d|)}$ as
  \begin{align}
    \Jb_0 \defeq
    \begin{bmatrix}
      \Lb_{11}\Cb_{11}\\
      \Lb_{21}\Cb_{21}\\
      \vdots\\
      \Lb_{dd}\Cb_{dd}
    \end{bmatrix}.
    \label{eq:def:Jb0}
  \end{align}
  Note that for all pairs of $(i, j)$ and $(s, t)$, we have $\Cb_{ij}\Cb_{st}^{\T} = \dtm_{ij;st}$ with $\dtm_{ij;st}$ as defined in \eqref{eq:dtm:4:entry}. Therefore, we obtain
  \begin{align}
    \Jb_0\Jb_0^{\T} = \Jb
    \label{eq:J:Jt}
  \end{align}
  with $\Jb$ as defined in \eqref{eq:Jb:def}.
  
  Then
  \begin{align}
    \left(\psib^{(i)}\right)^{\T} \dtmdmat \psib^{(j)}
    &= \sum_{s = 1}^d \sum_{t = 1}^d  \left(\psib_s^{(i)}\right)^{\T} \dtmdmat_{st} \psib_t^{(j)}\notag\\
    &= \sum_{s = 1}^d \sum_{t = 1}^d  \trop{\psib_t^{(j)}\left(\psib_s^{(i)}\right)^{\T} \dtmdmat_{st} }\notag\\
    &= \sum_{s = 1}^d \sum_{t = 1}^d  \vecop^{\T}(\dtmdmat_{st}) \vecop\left(\psib_s^{(i)}\left(\psib_t^{(j)}\right)^{\T}\right)\notag\\
    &= \sum_{s = 1}^d \sum_{t = 1}^d  \vecop^{\T}(\dtmdmat_{st}) \left(\psib_t^{(j)}\otimes\psib_s^{(i)}\right)\notag\\
    &=   \zetab^{\T} \left(\psib^{(j)}\circ\psib^{(i)}\right),
    \label{eq:vec:Xib:0}
  \end{align}
  where the penultimate equality follows from the fact that $\vecop(\vb\ub^{\T}) = \ub \otimes \vb$, and the last equality follows from \eqref{eq:tracy}.

 %
  %
  %
  %
  %
  %

%
  
  Therefore, we can rewrite the objective function of \eqref{eq:opt:equiv} as
  \begin{subequations}
    \begin{align}
      &\sum_{i = 1}^k\sum_{j = k + 1}^{m} \frac{\left[\left(\psib^{(i)}\right)^{\T} \dtmdmat \psib^{(j)}\right]^2}{\lambda^{(i)} - \lambda^{(j)}}\notag\\
      &= \sum_{i = 1}^k\sum_{j = k + 1}^{m} \frac{\left[\zetab^{\T} \left(\psib^{(j)}\circ\psib^{(i)}\right)\right]^2}{\lambda^{(i)} - \lambda^{(j)}}\label{eq:opt:equiv:obj:1}\\
      &= \zetab^{\T}\left[\sum_{i = 1}^k\sum_{j = k + 1}^{m} \frac{\left(\psib^{(j)}\circ\psib^{(i)}\right)\left(\psib^{(j)}\circ\psib^{(i)}\right)^{\T}}{\lambda^{(i)} - \lambda^{(j)}}\right]\zetab\label{eq:opt:equiv:obj:2}\\
      &= \zetab^{\T}\Gb_k\zetab\label{eq:opt:equiv:obj:3}\\
      &= \xib_{X^d}^{\T} \Jb_0^{\T} \Gb_k \Jb_0 \xib_{X^d},\label{eq:opt:equiv:obj:4}
    \end{align}
  \end{subequations}
  where \eqref{eq:opt:equiv:obj:1} follows from \eqref{eq:vec:Xib:0}, \eqref{eq:opt:equiv:obj:3} follows from the definition \eqref{eq:gb:k} of $\Gb_k$, and \eqref{eq:opt:equiv:obj:4} follows from \eqref{eq:zetab:xib}.

  As a result, the optimal value of the optimization problem \eqref{eq:opt:equiv} is $\spectral{\Jb_0^{\T} \Gb_k \Jb_0}$, i.e., the largest singular value of the matrix $\Jb_0 ^{\T} \Gb_k \Jb_0$.

  Finally, since $\spectral{\Ab^{\T}\Ab} = \spectral{\Ab\Ab^{\T}}$, we have
  \begin{align*}
    \bbspectral{\Jb_0^{\T} \Gb_k \Jb_0}
    &= \bbspectral{\left(\Gb_k^{\frac12}\Jb_0\right)^{\T} \Gb_k^{\frac12} \Jb_0}\\
    &= \bbspectral{ \Gb_k^{\frac12} \Jb_0\left(\Gb_k^{\frac12} \Jb_0\right)^{\T}}\\
    &= \bbspectral{\Gb_k^{\frac12} \Jb \Gb_k^{\frac12}},
  \end{align*}
  where we have used \eqref{eq:J:Jt} and the fact that $\Gb_k^{\frac12} = \left(\Gb_k^{\frac12}\right)^\T$.

\fi

\bibliographystyle{IEEEtran}
\bibliography{ref}

\end{document}